\documentclass[11pt]{article}
\usepackage{graphicx}
\usepackage{amsmath,amssymb,amsthm,amsfonts}
\newtheorem{theorem}{Theorem}[section]
\newtheorem{corollary}{Corollary}[section]
\newtheorem{lemma}{Lemma}[section]

\newtheorem{definition}{Definition}[section]
\newtheorem{example}{Example}[section]
\newtheorem{remark}{Remark}
\numberwithin{equation}{section}
\begin{document}

\title
  {\bf Twisted hierarchies associated with the generalized sine-Gordon equation\\
  }
\author   { Hui Ma \& Derchyi Wu \\
{}\\
Department of Mathematical Sciences, Tsinghua University\\
Beijing 100084, P.R. China\\ 
hma@math.tsinghua.edu.cn\\
{}\\
Institute of Mathematics, Academia Sinica\\
Taipei 10617, Taiwan\\ 
mawudc@math.sinica.edu.tw
}

\maketitle

\section*{{Abstract}}

Twisted $U$- and twisted $U/K$-hierarchies are soliton hierarchies introduced by Terng to find higher flows of the generalized sine-Gordon equation. Twisted $\frac {O(J,J)}{O(J)\times O(J)}$-hierarchies are among the most important classes of twisted hierarchies. In this paper, interesting first and higher flows of twisted $\frac {O(J,J)}{O(J)\times O(J)}$-hierarchies are explicitly derived, the associated submanifold geometry is investigated and a unified treatment of the inverse scattering theory is provided.

\section{Introduction}
The interaction between differential geometry and partial differential equations has been studied since the 19-th century and it can be found in the works of Lie, Darboux, Goursat, Bianchi, B$\ddot {\rm a}$cklund, and E. Cartan. One of the best known examples is the correspondence between surfaces of constant negative Gaussian curvature and the solutions of the sine-Gordon equation. The generalized sine-Gordon equation (GSGE) \cite{TT80},
\cite{Ter80} 
\begin{eqnarray}
	&\alpha\in O(n),&\label{E:GSGE1}\\
	&\partial_{x_j}\alpha_{ki}=\alpha_{kj}f_{ij},\ f_{ii}=0, & i\ne j, \label{E:GSGE2}\\
	&\partial_{x_j}f_{ij}+\partial_{x_i}f_{ji}+\sum_{k\ne i,\,j}f_{ik}f_{jk}=\alpha_{1i}\alpha_{1j},& i\ne j\label{E:GSGE3}\\
	&\partial_{x_k}f_{ij}=f_{ik}f_{kj},& i,j,k\ \textsl{ distinct}\label{E:GSGE4}
\end{eqnarray}
where $1\le i,\,j,\,k\le n$, defined by the Gauss-Codazzi equation for $n$-dimensional submanifolds in $\mathbb R^{2n-1}$ with constant sectional curvature $-1$, is a natural multidimensional differential geometic generalization of the sine-Gordon equation. 

The GSGE has B$\ddot {\rm a}$cklund transformations, permutability formula \cite{Ter80}, \cite{TU00}, a Lax pair, and an inverse scattering theory \cite{ABT86}. Recently, via a Lie algebra splitting approach, Terng introduced twisted $U$- and twisted $U/K$-hierarchies (twisted hierarchies on symmetric spaces)  \cite{T07} and showed that the GSGE can be interpreted as the $1$-dimensional system of the twisted $\frac{O(n,n)}{O(n)\times O(n)}$-hierarchy. As a result, she obtained all higher commuting flows of the GSGE.

There are rich intertwined analytic, algebraic and geometric structures of these twisted hierarchies. Twisted $\frac {O(J,J)}{O(J)\times O(J)}$-hierarchies are among the most important classes of twisted hierarchies which contains the twisted $\frac {O(n,n)}{O(n)\times O(n)}$-hierarchy ($J=I$) as a special case and possesses prototypical analytic, algebraic and geometric structures of twisted hierarchies. For instance, 
the associated Lax pairs of the twisted $\frac {O(J,J)}{O(J)\times O(J)}$-hierarchy are Laurent polynomials in the spectral variable $\lambda$. 
Therefore, the eigenfunctions depend on $x$ as $|\lambda|\to\infty$ and
we need to renormalize the eigenfunctions in solving the inverse scattering problem. The renormalization process breaks the symmetries. Thus a proper gauge is needed to reconstruct the symmetries and the potentials. 
Inspired by the result of \cite{ABT86}, we reconstruct the symmetry via solving an exterior product partial differential system which is derived from the associated $1$-dimensional system.

On the other hand, in the study of the submanifold geometry associated with twisted hierarchies, besides the case $J=I$ mentioned above  \cite{Ter80}, \cite{TT80}, \cite{ABT86}, we discover that the $1$-dimensional systems of the twisted $\frac {O(J,J)}{O(J)\times O(J)}$-hierarchy, $J=I_{1,n-1}$, describe the geometry of $n$-dimensional time-like submanifolds of constant positive sectional curvature in  $\mathbb R^{2n-1}_1$ . Therefore, rather than the space-like submanifolds associated with the generating equation introduced by Tenenblat \cite{CT94}, \cite{Ten98},  the time-like submanifolds interpret the geometry of soliton equations and can be tackled via an inverse scattering method. In this respect, twisted $\frac {O(J,J)}{O(J)\times O(J)}$-flows are more intimate with the GSGE or the sine-Gordon equation.

Finally, many interesting soliton hierarchies can be constructed from splittings of loop algebras fixed by involutions or automorphisms \cite{A79}, \cite{DS84}, \cite{W91}, \cite{TU00}. This observation motivates the classification theory of integrable systems via different representations and possible reductions \cite{M81}, \cite{FK83}, \cite{ML04}, \cite{ML05}, \cite{LS10}. The set up of the correspondence between the  reduction groups and the inverse scattering theory then becomes an important issue for mathematicians \cite{GGK01}, \cite{GG10}, \cite{GMV10}.  Since the twisted $\frac {O(J,J)}{O(J)\times O(J)}$-hierarchies are integrable equations induced by two involutions 
\[
\sigma_0(\xi(-\lambda))=\xi(\lambda),\quad \sigma_1(\xi(1/\lambda))=\xi(\lambda),
\]
on the loop group in the symmetric space $\frac {O(J,J)}{O(J)\times O(J)}$. Our work provides a complete inverse scattering theory of integrable  systems with the reduction group given by the dihedral group $D_2$. 

The paper is organized as follows: in Section \ref{S:TH}, we define the twisted $\frac {O(J,J)}{O(J)\times O(J)}$-hierarchies via splittings of loop algebras and compute explicit examples which include a new $4$-th order partial differential system (\ref{E:4-th-1}), (\ref{E:4-th-2}).  Section \ref{S:example} is devoted to the investigation of the associated submanifold geometry. In particular, we prove that the $1$-dimensional twisted $\frac {O(J,J)}{O(J)\times O(J)}$-system (twisted by $ \sigma_1$) with $J=I_{1,n-1}$  is the Gauss-Codazzi equation for an $n$-dimensional time-like  submanifold of constant sectional curvature $1$ in $\mathbb R^{2n-1}_1$ and derive a B$\ddot{\rm a}$cklund transformation theory for the $1$-dimensional twisted $\frac {O(J,J)}{O(J)\times O(J)}$-system. 
In Section \ref{S:invTH}, we solve  the direct problem by constructing special eigenfunctions which correspond to global twisted flows with nice decaying properties and regularities and  extracting the scattering data. Section \ref{S:inv-LGF} and \ref{S:inv-com} are devoted to the reconstruction of the flows from  scattering data. In particular, by studying the Riemann-Hilbert problem of the twisted $\frac {O(J,J)}{O(J)\times O(J)}$-flows, eigenfunctions with arbitrary poles and multiplicity are constructed for $J=I$, and eigenfunctions with small purely continuous scattering data are derived  for $J\ne I$.  The Cauchy problems of twisted flows and the $1$-dimensional twisted $\frac {O(J,J)}{O(J)\times O(J)}$-system  are solved in Section \ref{S:cauchy}.

\section{The twisted $\frac {O(J,J)}{O(J)\times O(J)}$-hierarchy}\label{S:TH}
To define the twisted $\frac {O(J,J)}{O(J)\times O(J)}$-hierarchy via a loop group approach, for an integer $q$, $0\le q \le n$, let us denote 
\begin{equation}
J=I_{q,n-q}=\mathrm{diag}(\overbrace{-1,\cdots, -1}^{q \ \mathrm{ times}},\overbrace{1,\dots,1}^{n-q \ \mathrm{ times}}),\quad \tilde J=\left(\begin{array}{cr} J & 0\\ 0 & -J\end{array}\right),
\label{E:J}
\end{equation}and  
\begin{eqnarray*}
O(J,J) &=&\left\{x\in GL_{2n}(\mathbb R)|\   x^t \tilde Jx=\tilde J\right\},\\
{ o }(J,J) &=& \{\xi\in gl_{2n}(\mathbb R)|\  \xi^t\tilde J+\tilde J\xi=0\}.
\end{eqnarray*}
For $i=0,\,1 $, let $\sigma_i$ be the involutions on $O(J,J)$ defined by
\[
\sigma_i(x)=I_{n+i,n-i}xI_{n+i,n-i}^{-1},
\]
and
\[
{ o }(J,J)=\mathcal K_i+\mathcal P_i
\]
the Cartan decompositions for $\sigma_i$. So $K_0\cap K_1=S_0\times S_1$, $K_0=S_0\times K_0'$, $K_1=K_1'\times S_1$ as direct product of subgroups with $\mathcal K_i$, $\mathcal S_i$, $\mathcal K_i'$ be the Lie algebras of $K_i$, $S_i$, $K_i'$. More precisely,
\begin{eqnarray*}
\mathcal K_0=o(q,n-q)+ o(q,n-q),&& \mathcal K_1=\left\{
                \begin{array}{ll}
                  o(n,1)+o(n-1), & \hbox{if } q=0, \\
                  o(q,n-q+1)+ o(q-1,n-q), & \hbox{if } q>0,
                \end{array}
              \right.\\
\mathcal S_0=o(q,n-q)+ 0_n,&& \mathcal S_1=\left\{
                \begin{array}{ll}
                  0_{n+1} + o(n-1), & \hbox{if } q=0, \\
                 0_{n+1} + o(q-1,n-q), & \hbox{if } q>0,
                \end{array}
              \right.\\
\mathcal K_0'=0_n + o(q,n-q),&& \mathcal K_1'=\left\{
                \begin{array}{ll}
                  o(n,1)+ 0_{n-1}, & \hbox{if } q=0, \\
                 o(q,n-q+1)+ 0_{n-1}, & \hbox{if } q>0.
                \end{array}
              \right.
\end{eqnarray*}

Let 
\begin{eqnarray*} 
L&=&\{f:\mathbb A_{\epsilon,1/\epsilon}\stackrel{holo.}{\rightarrow}GL_{2n}(\mathbb C)|\ \left(f(\bar\lambda)\right)^\ast \tilde Jf(\lambda)=\tilde J, \ \overline{f(\bar\lambda)}=f(\lambda)\}\\
L^{\sigma_0}  &=&\left\{f\in L|\  \sigma_0(f(-\lambda))=f(\lambda)\right\},\\
L^{\sigma_0}_{+}  &=& \{f\in L^{\sigma_0} | \sigma_1\left(f(1/\lambda)\right)=f(\lambda), \, f(1)\in  K'_1\}, \\
L^{\sigma_0}_{-}    &=&\{f\in  L^{\sigma_0}|f:\mathbb C/\mathbb D_{\epsilon}\stackrel{holo. }{\rightarrow}GL_{2n}(\mathbb C), \ f(\infty)\in  K_0'   \}.
\end{eqnarray*}
Here $\mathbb S^{\epsilon}$, $\mathbb S^{1/\epsilon}$ are circles of radius $\epsilon$, and $1/\epsilon\ne 1$ centered at $0$, $\mathbb D_{\epsilon}$ is the disk of radius $\epsilon\ne 1$, and $\mathbb A_{\epsilon,1/\epsilon}$ is the annulus with boundaries $\mathbb S^{\epsilon}$ and $\mathbb S^{1/\epsilon}$.  Then $L^{\sigma_0}_{+} \cap L^{\sigma_0}_{-} =\{1\}$ and the Lie algebras of $L^{\sigma_0}$, $ L^{\sigma_0}_{+}$, $L^{\sigma_0}_{-}$ are
\begin{eqnarray*}
&&{\mathcal L}^{\sigma_0} = \{\xi(\lambda)=\sum_{j\le n_0}\xi_j\lambda^j |\ \textit{$\xi_j\in\mathcal K_0$ if $j$ is even, $\xi_j\in\mathcal P_0$ if $j$ is odd}\},\\
&&{\mathcal L}^{\sigma_0}_{+} =  \{\xi\in \mathcal L^{\sigma_0} |\xi_{-j}=\sigma_1(\xi_j),\, \xi(1)\in\mathcal K'_1 \},\\
&&{\mathcal L}^{\sigma_0}_{-} =  \{\xi\in \mathcal L^{\sigma_0}|\xi(\lambda)=\sum_{j\le 0}\xi_j\lambda^j,\, \xi_0\in\mathcal K_0'  \}.
\end{eqnarray*}
It is a theorem that $({\mathcal L}^{\sigma_0}_{+}, {\mathcal L}^{\sigma_0}_{-})$ is a splitting of ${\mathcal L}^{\sigma_0}$ with  $\hat\pi_{\pm}$ defined by  
\begin{eqnarray}
&&\hat\pi_{+}(\xi) =\pi_{\mathcal S_0}(\xi_0)-\pi_{\mathcal S_1}\left(\sum_{0<j,\textit{ $j$ even}}(\,\xi_j+\sigma_1(\xi_j)\,)\right)+\sum_{0<j\le n_0}\left(\xi_j\lambda^j+\sigma_1(\xi_j)\lambda^{-j}\right)\label{E:pi+}\\
&&\hat\pi_{-}(\xi) =\pi_{\mathcal K_0'}(\xi_0)+\pi_{\mathcal S_1}\left(\sum_{0<j,\textit{ $j$ even}}(\,\xi_j+\sigma_1(\xi_j)\,)\right)+\sum_{0<j\le n_0}(\,\xi_{-j} -\sigma_1(\xi_j)\,)\lambda^{-j}\label{E:pi-}
\end{eqnarray}
as the projections of $\xi=\sum_{j\le n_0}\xi_j\lambda^j\in \mathcal L^{\sigma_0}$ onto $ {\mathcal L}_{\pm}^{\sigma_0}$ with respect to the splitting \cite{T07}.  
Finally, let $\mathcal A$ be a maximal abelian subalgebra of $\mathcal P_0$ consisting of elements of the form $\left(\begin{array}{cc} 0 & D\\ D &0\end{array}\right)$, where $D$ is a diagonal matrix in $gl(n,\mathbb R)$. Note $\sigma_1(\mathcal A)\subset \mathcal A$. Define
\begin{eqnarray}
&&	J_{a,2j+1}=a\lambda^{2j+1}+\sigma_1(a)\lambda^{-(2j+1)}\in  {\mathcal L}^{\sigma_0}_{+}\label{E:jaj}
\end{eqnarray}
for some constant regular $a\in\mathcal A$. 

\begin{definition}{\label{D:TH}}
The $2j+1$-th twisted ${\frac {O(J,J)}{O(J)\times O(J)}}$-flow (twisted by $\sigma_1$) is the compatibility condition of
\begin{eqnarray}
&&\left[\partial_x+\hat\pi_{+}\left(M J_{a,1}M^{-1}\right), \partial_t+\hat\pi_{+}\left(M J_{\tilde a,2j+1}M^{-1}\right)\right]=0,
\label{E:j-thflow}
\end{eqnarray}
for some $M=M(x,\lambda)\in  L^{\sigma_0}_{-}$.
\end{definition}


\begin{theorem}\label{T:pde}
Suppose $ a=\left(\begin{array}{cc} 0 &  D\\  D &0\end{array}\right)$, $\tilde a=\left(\begin{array}{cc} 0 & \tilde D\\ \tilde D &0\end{array}\right)$, $D=\textit{diag }(w_1,\cdots,w_n)$, $\tilde D=\textit{diag }(\tilde w_1,\cdots,\tilde w_n)$,  $ w_i\ne\pm w_j$, $ \tilde w_i\ne\pm \tilde w_j$, for $i\ne j$.
Then $2j+1$-th twisted ${\frac {O(J,J)}{O(J)\times O(J)}}$-flow is a nonlinear $2j+2$-th order partial differential system in the components of $b$, $v$, with
\begin{gather}
\hat\pi_+\left(M  J_{a,1}M^{-1}\right)=bab^{-1}\lambda+v+\sigma_1(bab^{-1})\frac 1\lambda, \qquad
b\in K_0',\   v\in\mathcal S_0.\label{E:pi}
\end{gather}
\end{theorem}

The proof of the above theorem follows from Lemma \ref{L:gaugeM}-\ref{L:Acoeff}.
\begin{lemma}\label{L:gaugeM}
The loop $M$ in (\ref{E:j-thflow}) can be chosen to satisfy
\begin{gather*}
M(\partial_x+J_{a,1})M^{-1}=\partial_x+\hat\pi_+\left(M  J_{a,1}M^{-1}\right),\\
M(\partial_t+J_{\tilde a,2j+1})M^{-1}=\partial_t+\hat\pi_+\left(M  J_{\tilde a,2j+1}M^{-1}\right).
\end{gather*}
\end{lemma}
\begin{proof}
Let $\Psi(x,t,\lambda)$ satisfy
\begin{eqnarray*}
&&\partial_x\Psi=-\hat\pi_{+}\left(M J_{a,1}M^{-1}\right)\Psi,\\
&&\partial_t\Psi=-\hat\pi_{+}\left(M J_{\tilde a,2j+1}M^{-1}\right)\Psi.
\end{eqnarray*}
Write $\Psi(x,t,\lambda)=m(x,t,\lambda)e^{-x(\lambda a+\frac 1\lambda\sigma_1(a))-t(\lambda^{2j+1} \tilde a+\frac 1{\lambda^{2j+1}}\sigma(\tilde a))}$. Then we derive 
\begin{eqnarray}
m(\partial_x+J_{a,1})m^{-1}&=&\partial_x+\hat\pi_+\left(M  J_{a,1}M^{-1}\right),\label{E:gauge1}\\
m(\partial_t+J_{\tilde a,2j+1})m^{-1}&=&\partial_t+\hat\pi_+\left(M  J_{\tilde a,2j+1}M^{-1}\right).\label{E:gauge1-1}
\end{eqnarray}
Hence
\begin{eqnarray}
\hat\pi_+\left(mJ_{a,1}m^{-1}\right)&=&\hat\pi_+\left(M  J_{a,1}M^{-1}\right),\label{E:gauge3}\\
\hat\pi_+\left(mJ_{\tilde a,2j+1}m^{-1}\right)&=&\hat\pi_+\left(M  J_{\tilde a,2j+1}M^{-1}\right),\label{E:gauge3-1}
\end{eqnarray}
by taking the projection $\hat\pi_+$ on both sides of (\ref{E:gauge1}), (\ref{E:gauge1-1}). Plugging (\ref{E:gauge3}), (\ref{E:gauge3-1}) into (\ref{E:gauge1}), (\ref{E:gauge1-1}), we then have
\begin{eqnarray*}
m(\partial_x+J_{a,1})m^{-1}&=&\partial_x+\hat\pi_+\left(m  J_{a,1}m^{-1}\right),\\
m(\partial_t+J_{\tilde a,2j+1})m^{-1}&=&\partial_t+\hat\pi_+\left(m  J_{\tilde a,2j+1}m^{-1}\right).
\end{eqnarray*}
The property $m(x,t,\lambda)\in L^{\sigma_0}_-$ will be shown in Theorem \ref{T:existence}.
\end{proof}

Define the $\lambda$-coefficients of $\hat\pi_+(MJ_{\tilde a,2j+1}M^{-1})$, $\hat\pi_-(MJ_{\tilde a,2j+1}M^{-1})$ by
\begin{eqnarray}
&&\hat\pi_+(MJ_{\tilde a,2j+1}M^{-1})=\sum_{s=1}^{2j+1}Q_s(x)\lambda^s +Q_0(x)+\sum_{s=1}^{2j+1}\sigma_1(Q_s)\lambda^{-s},\label{E:coeffM}\\
&&\hat\pi_-(MJ_{\tilde a,2j+1}M^{-1})=R_0(x)+\sum_{s>0}R_s(x)\lambda^{-s}\label{E:coeffM-1}
\end{eqnarray}
by (\ref{E:pi+}), (\ref{E:pi-}).

\begin{lemma}\label{L:recursive}
Let $I$ be the $n\times n$ identity matrix, $U=\frac 1{\sqrt 2}\left(\begin{array}{cc}I &-I\\ I &I\end{array}\right)$,  and
\begin{eqnarray*}
q_i(x)&=&(bU)^{-1}Q_{i}bU,\quad 0\le i\le 2j+1,\\
r_0(x)&=&(bU)^{-1}R_0bU,\\
f(x,\lambda)&=&(bU)^{-1}MJ_{\tilde a,2j+1}M^{-1}bU\\
&=& \sum_{s=1}^{2j+1}q_s\lambda^s +q_0+r_0+(bU)^{-1}\left(\sum_{s=1}^{2j+1}\sigma_1(Q_s)\lambda^{-s}+\sum_{s>0}R_s\lambda^{-s}\right)\,bU.
\end{eqnarray*}
Then 
\begin{eqnarray}
 q_{2j+1}&=&\ U^{-1}\tilde aU,\nonumber\\
\left[U^{-1}aU,q_{2j}\right]&=&-(bU)^{-1}\left(\partial_x Q_{2j+1}\right)bU-\left[(bU)^{-1}vbU,q_{2j+1}\right],\label{E:2j}\\
\left[U^{-1}aU,q_{s\ }\right]&=&-(bU)^{-1}\left(\partial_x Q_{s+1\ }\right)bU-\left[(bU)^{-1}vbU,q_{s+1}\right]\label{E:s}\\
&&-\left[(bU)^{-1}\sigma_1(bab^{-1})bU,q_{s+2}\right],\quad 1\le s\le 2j-1,\,s\ne 2j,\,0\nonumber\\
\left[U^{-1}aU,q_0+r_0\right]&=&-(bU)^{-1}\left(\partial_x Q_{1\ }\right)bU-\left[(bU)^{-1}vbU,q_{1}\right]\label{E:0}\\
&&-\left[(bU)^{-1}\sigma_1(bab^{-1})bU,q_{2}\right].\nonumber
\end{eqnarray}
Moreover,  
\begin{eqnarray}
&&(f-\tilde w_1(\lambda^{2j+1}+\frac 1{\lambda^{2j+1}}))\prod_{s=2}^n \left(f-\tilde w_s(\lambda^{2j+1}-\frac 1{\lambda^{2j+1}})\right)\label{E:characteristic}\\
\cdot&&(f+\tilde w_1(\lambda^{2j+1}+\frac 1{\lambda^{2j+1}}))\prod_{s=2}^n \left(f+\tilde w_s(\lambda^{2j+1}-\frac 1{\lambda^{2j+1}})\right)=0.\nonumber
\end{eqnarray}
\end{lemma}
\begin{proof} The identity (\ref{E:characteristic}) follows from the characteristic polynomial of $f$. By Lemma \ref{L:gaugeM}
, we obtain
\[
\left[\partial_x+ \hat\pi_+\left(M  J_{a,1}M^{-1}\right), M  J_{\tilde a,{2j+1}}M^{-1}\right]=0,
\]Therefore by (\ref{E:pi}), and (\ref{E:coeffM}), we derive
\begin{eqnarray}
 Q_{2j+1}&=&b\tilde ab^{-1},\nonumber\\
\left[bab^{-1},Q_{2j}\right]&=&-\partial_x Q_{2j+1}-\left[v,Q_{2j+1}\right],\nonumber\\
\left[bab^{-1},Q_{s\ }\right]&=&-\partial_x Q_{s+1\ }-\left[v,Q_{s+1}\right]-\left[\sigma_1(bab^{-1}),Q_{s+2}\right],\ 1\le s\le 2j-1,\nonumber\\
\left[bab^{-1},Q_{0}+R_0\right]&=&-\partial_x Q_{1}-\left[v,Q_{1}\right]-\left[\sigma_1(bab^{-1}),Q_{2}\right].\nonumber
\end{eqnarray}
Hence follows the lemma.
\end{proof}

\begin{lemma}\label{L:Acoeff}
For $\forall 0\le s\le 2j$, the entries of 
$Q_s$ are fixed functions of components of $\partial_x^{\alpha} b$ and $\partial_x^{\beta} v$, $0\le \alpha, \,\beta \le 2j-s+1$. 
\end{lemma}
\begin{proof} 
First of all, write 
$q_s=T_s+P_s$, for $s>0$ and $q_0+r_0=T_0+P_0$ with $T_i$, $P_i$ being diagonal and off-diagonal respectively. Note that $U^{-1}aU$ is a diagonal matrix. Hence $P_{2j}$ can be derived in terms of $b$, $b_x$, $v$ by (\ref{E:2j}) once $ w_i\ne\pm w_j$.
Equating the $\lambda^{(2j+1)({2n})-1}$-coefficients of the diagonal part of (\ref{E:characteristic}), we conclude $T_{2j}=0$. Hence  the lemma is done if $j=0$. As a result,  we can assume $j>0$ in the following proof.

We are going to prove the lemma for $0\le s\le 2j-1$ by induction. Similarly, $P_{s}$ can be derived in terms of $\partial_x^\alpha b$, $\partial_x^\beta v$, $0\le\alpha,\, \beta \le 2j-s+1$, by (\ref{E:s}), (\ref{E:0}) and the induction hypothesis. Using  (\ref{E:characteristic})  and equating the $\lambda^{(2j+1)({2n})-(2j+1-s)}$-coefficients of the diagonal part of  (\ref{E:characteristic}), we obtain 
\[
 T_{s}{\textit{diag}}(\prod_{1\le k\le 2n,\,k\ne 1}(\lambda_1-\lambda_k),  \cdots,\prod_{1\le k\le 2n,\,k\ne 2n}(\lambda_{2n}-\lambda_k))=F_s(T_\alpha, P_\beta),
\]
Here ${\textit{diag }}(\lambda_1, \lambda_2,\cdots,\lambda_{2n})={\textit{diag }}(\tilde w_1,\cdots, \tilde w_n, -\tilde w_1,\cdots, -\tilde w_n)$, the entries of $F_s$ are fixed polynominal functions of those of $T_\alpha$, $P_\beta$, and $s+1\le\alpha\le 2j+1, \, s\le\beta\le 2j+1$. 
Therefore, the lemma is justified if $ \tilde w_i\ne\pm \tilde w_j$.
\end{proof}

Via the algorithm provided in the proof of Theorem \ref{T:pde} and the Maple $7$ software, we derive
\begin{example}\label{Ex:3rd} 
For $(n,q)=(2,0)$, $a=\tilde a$ ($w_1=\tilde w_1=1$, $w_2=\tilde w_2=2$), \textbf{the first flow} is the trivial linear system
\[
\partial_t u=\partial_x u,\quad \partial_t \omega=\partial_x \omega,
\]
where
\begin{eqnarray}
b&= &\left(\begin{array}{cccc}
1 &0&0& 0\\
0&1&0&0\\
0 &0&\cos u(x,t)& \sin u(x,t)\\
0&0&-\sin u(x,t)& \cos u(x,t)
\end{array}\right),\label{E:hier-b}\\
v&=& \left(\begin{array}{cccc}
0& -\omega(x,t)&0&0\\
\omega(x,t)& 0&0&0\\
0 & 0&0&0\\
0&0&0&0
\end{array}\right),\label{E:hier-v'}
\end{eqnarray}
and \textbf{the third flow} of the twisted $\frac {O(J,J)}{O(J)\times O(J)}$-hierarchy (twisted by $\sigma_1$) is the $4$th order partial differential system
\begin{eqnarray}
\partial_t u &=&\frac  1{18}\{\ 10\partial_x^3u+\partial_xu\left[5(\partial_xu)^2-12\omega\partial_xu+180\left(\cos u\right)^2-90+15\omega^2\right]\label{E:4-th-1}\\
&&-8\partial_x^2\omega-4\omega^3+\omega\left[24-48\left(\cos u\right)^2\right]\ \},\nonumber\\
\partial_t \omega &=&-\frac  49\partial_x^4u+\partial_x^2u[-\frac 23\left(\partial_x u\right)^2+\frac 53\omega\partial_xu-\frac {40}3\left(\cos u\right)^2+\frac{20}3-\frac 23 \omega^2 ]   \label{E:4-th-2}  \\
&&-32\sin u\left(\cos u\right)^3+16\cos u\sin u\nonumber\\
&&+\partial_xu[\frac {40}3\left(\partial_x u\right)\cos u\sin u+\frac 56\left(\partial_xu\right)\left(\partial_x\omega\right)-\frac 43\omega\partial_x\omega]\nonumber\\
&&+\frac 59\partial_x^3 \omega+\partial_x\omega[\frac 56\omega^2-5+10\left(\cos u\right)^2]-\frac 83\omega^2\cos u\sin u.\nonumber
\end{eqnarray}
Here the associated Lax pair (\ref{E:j-thflow}) is
\[
\left[\partial_x+bab^{-1}\lambda+v+\sigma_1(bab^{-1})\frac 1\lambda, \partial_t+\sum_{s=1}^3Q_s\lambda^s+ Q_0+\sum_{s=1}^3\sigma_1(Q_s)\lambda^{-s}\right]=0,
\]with $b$, $v$ defined by (\ref{E:hier-b}), (\ref{E:hier-v'}), 
\begin{eqnarray*}
Q_3 &=&\left(\begin{array}{cccc}
0&0&\cos u & -\sin u\\
0&0&2\sin u & 2\cos u\\
\cos u & 2\sin u&0&0\\
-\sin u & 2\cos u&0&0
\end{array}
\right),\\
Q_2 &=&\left(\begin{array}{cccc}
0 & -\omega&0&0\\
\omega & 0&0&0\\
0&0&0 & \partial_x u\\
0&0&-\partial_x u & 0
\end{array}
\right),\\
Q_1 &=&\left(\begin{array}{cccc}
0&0&\alpha_{11} & \alpha_{12}\\
0&0&\alpha_{21}& \alpha_{22}\\
\alpha_{11} & \alpha_{21}&0&0\\
\alpha_{12} & \alpha_{22}&0&0
\end{array}
\right),\\
Q_0& =&\left(\begin{array}{cccc}
0 & -\beta&0&0\\
\beta & 0&0&0\\
0&0&0 & 0\\
0&0&0 & 0
\end{array}
\right),
\end{eqnarray*}
and
\begin{eqnarray*}\alpha_{11}&=&\frac 13 \left(\sin u\right)\partial_x^2u+\partial_xu\left[-\frac 16\left(\cos u\right)\partial_x u+\frac 23\omega\cos u\right]\\
&&-2\left(\cos u\right)^3+2\cos u-\frac 16\omega^2\cos u-\frac 23\left(\sin u\right)\partial_x\omega,\\
\alpha_{12}&=&\frac 13 \left(\cos u\right)\partial_x^2u+\partial_xu\left[\frac 16\left(\sin u\right)\partial_x u-\frac 23\omega\sin u\right]\\
&&+\frac 16\omega^2\sin u+2\left(\cos u\right)^2\sin u-\frac 23\left(\cos u\right)\partial_x\omega,
\\
\alpha_{21}&=&\frac 23 \left(\cos u\right)\partial_x^2u+\partial_xu\left[\frac 13\left(\sin u\right)\partial_x u-\frac 13\left(\sin u\right)\omega\right]\\
&&+\frac 13\omega^2\sin u+4\left(\cos u\right)^2\sin u-\frac 13\left(\cos u\right)\partial_x\omega,\\
\alpha_{22}&=&-\frac 23 \left(\sin u\right)\partial_x^2u+\partial_xu\left[\frac 13\left(\cos u\right)\partial_x u-\frac 13\omega\cos u\right]\\
&&+4\left(\cos u\right)^3-4\cos u+\frac 13\omega^2\cos u+\frac 13\left(\sin u\right)\partial_x\omega,
\\
\beta &=& -\frac 49\partial_x^3u+\partial_xu\left[-\frac 29\left(\partial_xu\right)^2+\frac 56\omega\partial_xu-8\left(\cos u\right)^2+4-\frac 23\omega^2\right]\\
&&+\frac 59\partial_x^2\omega-\frac 53\omega+\frac 5{18}\omega^3+\frac {10}3 \omega\left(\cos u\right)^2.
\end{eqnarray*}
\end{example}

\begin{example}\label{Ex:1st-nonlinear}
If  $(n,q)=(2,0)$, and define $a$, $\tilde a$ by $w_1=1$, $w_2=2$, $\tilde w_1=2$, $\tilde w_2=1$, $b$, $v'$ by (\ref{E:hier-b}),  (\ref{E:hier-v'}), then the first flow is the sine-Gordon equation
\[
\partial_t^2 u-\partial_x^2 u=12\sin u\cos u,\quad \partial_t u= \omega.
\]
\end{example}

\begin{remark}\label{R:sigmar}
For $i\in\{0,\,1,\,\cdots,\,n-1\}$,  replacing $\sigma_1$ by $\sigma_i$, $
\sigma_i(x)=I_{n+i,n-i}xI_{n+i,n-i}^{-1}$, and $\sigma_1(f(1/\lambda))$ either by $
\sigma_i(f(1/\lambda))$ or by $\sigma_i(f(-1/\lambda))$, 
we can generalize the definition of twisted hierarchies by analogy.
\end{remark}

\section{The $1$-dimensional system}\label{S:example}
We discuss associated submanifold geometry of twisted $\frac {O(J,J)}{O(J)\times O(J)}$-flows. 
A $1$-dimensional system is constructed by putting all first flows together in a soliton hierarchy. 
Many $1$-dimensional systems are the Gauss-Codazzi equations for submanifolds in space forms or symmetric spaces with special geometric properties. For instance, the Gauss-Codazzi equations for isothermic surfaces in $\mathbb R^3$ is the $1$-dimensional system of the $\frac {O(4,1)}{O(3)\times O(1,1)}$-hierarchy \cite{CGS95}, \cite{BDPT02}, \cite{B06}. Other interesting examples can be found in  \cite {T07}. Similarly, for the $1$-dimensional twisted $\frac {O(J,J)}{O(J)\times O(J)}$-system, one has:
\begin{definition}\label{D:1Dsystem}
The $1$-dimensional twisted $\frac {O(J,J)}{O(J)\times O(J)}$-system 
 (twisted by $\sigma_1$) is the compatibility condition of
\begin{gather}
\left[\partial_{x_i}+\hat\pi_{+}\left(M  J_{a_i,1}M^{-1}\right), \partial_{x_j}+\hat\pi_{+}\left(M  J_{a_j,1}M^{-1}\right)\right]=0,\ 1\le i,\,j\le n
\label{E:1Dsystem1}
\end{gather}
for some $M=M(x_1,\cdots,x_n,\lambda)\in L^{\sigma_0}_{-}$,  and 
\begin{equation}
a_i=
\left(\begin{array}{lr}
0 & e_i
\\
e_i & 0
\end{array}\right),\quad e_i=\textsl{diag}(0,\cdots,0,\stackrel{i-th\, entry}{1},0,\cdots,0)\in gl(n,\mathbb C).\label{E:ai}
\end{equation}
\end{definition}

\begin{example}\label{Ex:sin} (\textbf{The sine-Gordon equation}) 
The $1$-dimensional twisted $\frac {O(J,J)}{O(J)\times O(J)}$-system (twisted by $\sigma_1$) with $(n,q)=(2,0)$ is the sine-Gordon equation.  \end{example}
\begin{proof}
In this casse, we have $j=0$, $J $ in (\ref{E:J}) is the $2\times 2$ identity matrix, and
\begin{eqnarray*}
O(J,J) = O(2,2),& &o(J,J) = o(2,2),\\
K_0=O(2)\times O(2),& &K_1=O(2,1)\times 1_1,\\
\mathcal S_0=o(2)+ 0_2,& & K_0'=	1_2\times O(2).
\end{eqnarray*}
Let 
\begin{gather}
a=a_1=\left(\begin{array}{cc}
0 & {\begin{array}{cc}
1& 0\\
0&0
\end{array}}\\
{\begin{array}{cc}
1& 0\\
0&0
\end{array}} & 0
\end{array}\right), \ 
\tilde a=a_2=\left(\begin{array}{cc}
0 & {\begin{array}{cc}
0& 0\\
0& 1
\end{array}}\\
{\begin{array}{cc}
0& 0\\
0& 1
\end{array}} & 0
\end{array}\right)\label{E:assump}
\end{gather}
in the Lax pair (\ref{E:j-thflow}) and write
\begin{eqnarray}
&&v_1=
\left(\begin{array}{cc}
{\begin{array}{cr}
0& -\alpha\\
\alpha & 0
\end{array}}& 0\\
0 & 0
\end{array}\right),\quad v_2=
\left(\begin{array}{cc}
{\begin{array}{cr}
0& -\beta\\
\beta & 0
\end{array}}& 0\\
0 & 0
\end{array}\right)\in \mathcal S_0,\label{E:v}
\\
{}&&\nonumber\\
&&b= \left(\begin{array}{cc}
1 & 0\\
0 &{\begin{array}{rr}
\cos\frac u2& \sin\frac u2\\
-\sin\frac u2& \cos\frac u2
\end{array}} 
\end{array}\right)\hskip.5in\in K_0'\label{E:b}
\end{eqnarray}
in (\ref{E:pi}) and 
$\hat\pi_+\left(M  J_{ a_2,1}M^{-1}\right)=ba_2b^{-1}\lambda+v_2+\sigma_1(ba_2b^{-1})\frac 1\lambda$. 
 Equating the $\lambda$-coefficients of (\ref{E:j-thflow}) , we then derive
\begin{gather}
\partial_t(ba_1b^{-1})-\partial_x(ba_2b^{-1})+	\left[v_2,ba_1b^{-1}\right]-\left[v_1,ba_2b^{-1}\right]=0,\label{E:cc1}\\
\partial_t v_1-\partial_x v_2-	\left[v_1,v_2\right]-\left[\sigma_1(ba_1b^{-1}),ba_2b^{-1}\right]+\left[\sigma_1(ba_2b^{-1}),ba_1b^{-1}\right]=0.\label{E:cc0}
\end{gather}
Plugging (\ref{E:v}), (\ref{E:b}) into (\ref{E:cc1}), (\ref{E:cc0}), we obtain
\begin{eqnarray*}
&&\alpha=\frac 12 \partial_t u,\qquad 
\beta=\frac 12 \partial_x u,\qquad
\partial_x \beta-\partial_t \alpha=2\sin u.
\end{eqnarray*}
Hence the $1$-dimensional twisted $\frac {O(2,2)}{O(2)\times O(2)}$-system (twisted by $\sigma_1$) is the sine-Gordon equation
\begin{equation}
u_{xx}-u_{tt}=4\sin u.\label{E:SG}
\end{equation}
\end{proof}

\begin{example}\label{Ex:sinh} (\textbf{The sinh-Gordon equation}) The $1$-dimensional twisted $\frac {O(J,J)}{O(J)\times O(J)}$-system (twisted by $\sigma_1$) with $(n,q)=(2,1)$ is the sinh-Gordon equation.

\end{example}
\begin{proof} We have $j=0$, $ J=\left(\begin{array}{cc}-1&0\\0&1\end{array}\right)$, and
\begin{eqnarray*}
K_0=O(1,1)\times O(1,1),& &K_1=O(1,2)\times 1_1,\\
\mathcal S_0=o(1,1)\times 0_2,& &K_0'=	1_2\times O(1,1).
\end{eqnarray*}
Let 
\begin{eqnarray*}
&&a=a_1=\left(\begin{array}{cc}
0 & {\begin{array}{cc}
1& 0\\
0&0
\end{array}}\\
{\begin{array}{cc}
1& 0\\
0&0
\end{array}} & 0
\end{array}\right), \ 
\tilde a=a_2=\left(\begin{array}{cc}
0 & {\begin{array}{cc}
0& 0\\
0& 1
\end{array}}\\
{\begin{array}{cc}
0& 0\\
0& 1
\end{array}} & 0
\end{array}\right)
\\
&&v_1=
\left(\begin{array}{cc}
{\begin{array}{cr}
0& \alpha\\
\alpha & 0
\end{array}}& 0\\
0 & 0
\end{array}\right),\quad v_2=
\left(\begin{array}{cc}
{\begin{array}{cr}
0& \beta\\
\beta & 0
\end{array}}& 0\\
0 & 0
\end{array}\right)\in \mathcal S_0=o(1,1)\times 0,
\\
{}&&\\
&&b= \left(\begin{array}{cc}
1 & 0\\
0 &{\begin{array}{rr}
\cosh\frac u2& \sinh\frac u2\\
\sinh\frac u2& \cosh\frac u2
\end{array}} 
\end{array}\right)\hskip.5in\in K_0'=1\times O(1,1).
\end{eqnarray*}
So a similar argument yields
\begin{eqnarray*}
&&\alpha=-\frac 12 \partial_t u,\qquad 
\beta=-\frac 12 \partial_x u,\qquad
\partial_x \beta-\partial_t \alpha=-2\sinh u
\end{eqnarray*}
and
\begin{equation}
u_{xx}-u_{tt}=4\sinh u.\label{E:ShG}
\end{equation}
\end{proof}

\begin{example}\label{Ex:twisteduk1} (\textbf{The generalized sine-Gordon equation}) 
The $1$-dimensional twisted $\frac {O(J,J)}{O(J)\times O(J)}$-system  (twisted by $ \sigma_1$) with $(n,q)=(n,0)$ is the Gauss-Codazzi equation for an $n$-dimensional submanifold of constant sectional curvature $-1$ in $\mathbb R^{2n-1}$, i.e. the generalized sine-Gordon equation (GSGE).
\end{example}
\begin{proof} We give an alternative proof (cf  \cite{T07}). The Gauss-Codazzi equation for an $n$-dimensional submanifold in $\mathbb R^{2n-1}$ of constant sectional curvature $-1$ is the generalized sine-Gordon equation (GSGE) (\ref{E:GSGE1})-(\ref{E:GSGE4}). 
Moreover, the B$\ddot {\rm a}$cklund transformation for the GSGE is constructed by showing that
\begin{eqnarray}
	&& dX-X\omega=D_\lambda \alpha \delta-X\delta \alpha^t D_\lambda  X	\label{E:Backlund}
\end{eqnarray}
gives a new solution to the GSGE if $\alpha$ is a given solution of the GSGE. Here
\begin{eqnarray*}
	&& \omega=\delta F-F^t \delta ,\qquad \delta=\sum_{j=1}^n e_jdx_j,\qquad D_\lambda =\frac 12(\lambda I-\frac 1\lambda I_{1,n-1})
\end{eqnarray*}
$F=(f_{ij})$ is defined by (\ref{E:GSGE2}), $e_j$ is defined as in (\ref{E:ai}),  and $I$ is the $n\times n$ identity matrix \cite{TT80}, \cite{Ter80}, \cite{ABT86}. Later the B$\ddot {\rm a}$cklund transformation (\ref{E:Backlund}) is linearized (so is the GSGE (\ref{E:GSGE1})-(\ref{E:GSGE4})) by the Lax pair
\begin{gather}
\partial_{x_j}\tilde \Psi = \left(\frac 12\lambda A_j+\frac 1{2\lambda}B_j+C_j\right)\tilde\Psi,\label{E:Lax0}
\end{gather}
with $A_j = \left(
\begin{array}{cc}
0 & \alpha e_j\\
(\alpha e_j)^t & 0
\end{array}
\right)$, $B_j = \left(
\begin{array}{cc}
0 & -I_{1,n-1}\alpha e_j\\
-(\alpha e_j)^tI_{1,n-1} & 0
\end{array}
\right)$, $C_j=\left(
\begin{array}{lr}
0 & 0\\
0 & \gamma_j
\end{array}
\right)$
, $\gamma_j=C(\frac{\partial}{\partial x_j})\in o(n)$, and solved by the inverse scattering method \cite{ABT86}. By a permutation, 
\begin{gather*}
Y\in gl(2n,\mathbb C)\quad\longmapsto \quad PYP^t,\\
P=\left(
\begin{array}{cc}
0 &\iota\\
I & 0
\end{array}\right),\quad \iota (e_{n-j})=e_{j+1},\ 0\le j\le n-1,
\end{gather*}and a change of coordinates
\begin{eqnarray*}
x_{n-j}&\quad\longmapsto \quad &  2x_{j+1},\ 1\le j\le n,\\
\lambda &\quad\longmapsto \quad &  -\lambda,
\end{eqnarray*}
the Lax pair (\ref{E:Lax0}) can be written as
\[
\partial_{x_j}\Psi = -( ba_jb^{-1}\lambda+v_j+\sigma_1(ba_jb^{-1})\frac 1\lambda\ )\Psi,
\]
where $a_j$ is defined by (\ref{E:ai}), and
\begin{eqnarray*}
&b=\left(
\begin{array}{cc}
I & 0\\
0 & g
\end{array}\right)\in K_0'=1\times O(n),& g=  \alpha \iota ,\\
&v_j=\left(
\begin{array}{cc}
u_j & 0\\
0 & 0
\end{array}\right)\in\mathcal S_0=o(n)\times 0,& u_j=-2\iota  \gamma_j\iota. 
\end{eqnarray*}
\end{proof}

\begin{example}(\textbf{The generalized sinh-Gordon equation})\label{Ex:twistedmw} 
The $1$-dimensional twisted $\frac {O(J,J)}{O(J)\times O(J)}$-system (twisted by $ \sigma_1$) with $(n,q)=(n,1)$  is the Gauss-Codazzi equation for an $n$-dimensional time-like  submanifold of constant sectional curvature $1$ in $\mathbb R^{2n-1}_1$ which possesses simultaneously diagonizable first and second fundamental forms.
\end{example}

To prove the statement of Example \ref{Ex:twistedmw}, we first show: 
\begin{theorem}\label{T:GC-sinhG}
Suppose $M$ is a time-like $n$-dimensional submanifold of constant sectional curvature $1$ in $\mathbb R^{2n-1}_1$. Suppose that there exist local coordinates $x_1,x_2,\cdots, x_n$ on a neighborhood of $p\in M$,
an $O(1,n-1)$-valued map $A=(a^i_j)$, and parallel normal frames $e_{n+1}, \cdots, e_{2n-1}$
such that
the first and second fundamental forms are of the form
\begin{equation}\label{E:I-II}
\mathrm{I}=\sum_{i=1}^n \epsilon_i (a^1_i)^2 dx_i\otimes dx_i,
\quad \mathrm{II}=\sum_{\lambda=2}^n\sum_{i=1}^{n} \epsilon_i a^1_i a^{\lambda}_i dx_i\otimes dx_i e_{n+\lambda-1},
\end{equation}
where $\mathrm{diag}(\epsilon_1, \epsilon_2, \cdots, \epsilon_n)=\mathrm{diag}(-1, 1, \cdots, 1)={I}_{1,n-1}=J$. Then the Gauss-Codazzi equation and the structure equation is the $1$-dimensional twisted $\frac {O(J,J)}{O(J)\times O(J)}$-system (\ref{E:1Dsystem1})
with
\begin{equation}\label{E:gav}
\begin{split}
&\hat\pi_+\left(M  J_{a,i}M^{-1}\right)=ba_ib^{-1}\lambda+v_i+\sigma_1(ba_ib^{-1})\frac 1\lambda,\\
&g=b^{-1}=\left(
    \begin{array}{cc}
      \mathrm{I}_n & 0 \\
      0 & A^t \\
    \end{array}
  \right): \mathbb{R}^n \rightarrow K_0^{\prime}=\mathrm{I}\times O(1,n-1),\\
&a_i=\frac 12\left(
    \begin{array}{cc}
      0 & e_{i} \\
      e_{i} & 0 \\
    \end{array}
  \right),\quad \textit{$e_i$ defined by (\ref{E:ai}) }\\
& v_i=\left(
        \begin{array}{cc}
          u_i & 0 \\
          0 & 0 \\
        \end{array}
      \right): \mathbb{R}^n \rightarrow \mathcal{S}_0=o(1,n-1)\oplus 0, \quad 1\leq i\leq n.
\end{split}
\end{equation} Moreover, writing $A=\left(a^i_j\right)$, the $1$-dimensional twisted $\frac {O(J,J)}{O(J)\times O(J)}$-system can be written as
\begin{eqnarray}
	&A\in O(1,n-1),&\label{E:GSHGE1}\\
	&\partial_{x_j}a_{i}^k=a_{j}^kf_{ij},\ f_{ii}=0, & i\ne j, \label{E:GSHGE2}\\
	&\epsilon_j\partial_{x_j}f_{ij}+\epsilon_i\partial_{x_i}f_{ji}+\sum_{k\ne i,\,j}\epsilon_kf_{ik}f_{jk}=-a_{i}^1 a_{j}^1,& i\ne j\label{E:GSHGE3}\\
	&\partial_{x_k}f_{ij}=f_{ik}f_{kj},& i,j,k\ \textsl{ distinct}\label{E:GSHGE4}
\end{eqnarray}
where $1\le i,\,j,\,k\le n$. We call the nonlinear system (\ref{E:GSHGE1})-(\ref{E:GSHGE4}) the \textbf{generalized sinh-Gordon equation}.
\end{theorem}
\begin{proof}

$\underline{\textsl{Step 1: the Gauss-Codazzi equation}}$

To write down the Gauss-Codazzi equations for these immersions we set
\begin{eqnarray}
&&\omega^i= a^1_i dx_i, \nonumber\\
&& \omega_{i}^{n+\lambda-1}= \epsilon_i a^{\lambda}_i dx_i,\nonumber
\end{eqnarray}
where $1\leq i,\ j\leq n$, $2\leq \lambda \leq n$. Hence by the structure equations
\[
d\omega^i=\omega^j\wedge \omega^i_j, \quad \epsilon_i\omega^i_j+\epsilon_j \omega^j_i=0,\label{E:omega-ij}\\
\]$\omega^i_j=f_{ij}dx_i-\epsilon_i\epsilon_j f_{ji}dx_j$,
where
\begin{equation}\label{E:a1j}
f_{ij}=\left\{
        \begin{array}{ll}
          \frac{(a^1_i)_{x_j}}{a^1_j}, & i\neq j; \\
          0, & i=j.
        \end{array}
      \right.
\end{equation}
Set $F=(f_{ij})$ and $\delta=\sum_{j=1}^n e_jdx_j$. Then
\begin{equation*}
\omega=(\omega^{i}_{j})_{1\leq i,j\leq n}=\delta F-JF^t\delta J
\end{equation*}
is the Levi-Civita $o(1,n-1)$-connection of the induced pseudo-Riemannian metric $\mathrm{I}$. The Gauss-Codazzi equation and the structure equation give
\begin{equation}\label{E:alj}
\left\{
\begin{array}{l}
d\omega+\omega\wedge \omega=\delta A^t e_1\wedge (AJ \delta )=-\delta A^t e_1\wedge (JAJ \delta ),\\
(a^{\lambda}_i)_{x_j}= a^{\lambda}_j f_{ij}, \quad  1\leq i, j \leq n, i\neq j, 2\leq \lambda \leq n.
\end{array}
\right.
\end{equation}
Moreover, it follows from (\ref{E:a1j}), (\ref{E:alj}) that
\begin{equation}\label{eq:a^k_i_j}
(a^{k}_i)_{x_j}= a^{k}_j f_{ij}, \quad  1\leq i, j,k \leq n, \ i\neq j.
\end{equation}
On the other hand, since $A=(a^i_j)\in O(1,n-1)$,
$$\sum_j \epsilon_j (a^k_j)^2=\epsilon_k.$$
Taking differential with respect to $x_i$ on the above equality, we get
$$\epsilon_ia^k_i(a^k_i)_{x_i}=-\sum_{j\neq i} \epsilon_j a^k_j (a^k_j)_{x_i}.$$
It follows from \eqref{eq:a^k_i_j} that
\begin{equation}\label{eq:a^k_i_i}
(a^k_i)_{x_i}=-\epsilon_i \sum_{j\neq i} \epsilon_j a^k_j f_{ji}.
\end{equation}
Then \eqref{eq:a^k_i_j} and \eqref{eq:a^k_i_i} can be expressed as
$$A^{-1}dA=\delta F^t -J F\delta J.$$

Summarize, $A=(a_{j}^i)$ satisfies the following second order PDE system:
\begin{equation}
\left\{
\begin{array}{ll}
d\omega+\omega\wedge \omega=\delta A^t e_1 \wedge AJ\delta=-\delta A^t e_1\wedge (JAJ\delta),\\
A^{-1}dA=\delta F^t -J F\delta J,\\
\text{ where } \omega=\delta F-JF^t \delta J,
\end{array}
\right.\label{E:GC-1}
\end{equation}

$\underline{\textsl{Step 2: the $1$-dimensional twisted $\frac{O(J,J)}{O(J)\times O(J)}$-system}}$

Define 
\begin{eqnarray}
\theta_{\lambda}&=&\sum_{i=1}^n ((g^{-1}a_i g)\lambda + v_i +\sigma_1(g^{-1} a_i g)\lambda^{-1})dx_i\label{E:flat-1system-1}\\
&=&\frac{\lambda}{2}\left(
                                    \begin{array}{cc}
                                      0 & \delta A^t \\
                                     JAJ\delta  & 0 \\
                                    \end{array}
                                  \right)
 + \left(
     \begin{array}{cc}
       u & 0 \\
       0 & 0 \\
     \end{array}
   \right)
- \frac{\lambda^{-1}}{2}\left(
                          \begin{array}{cc}
                            0 & \delta A^t J \\
                            AJ\delta & 0 \\
                          \end{array}
                        \right),\label{E:flat-1system-2}
\end{eqnarray}
where  $u=\sum_{i=1}^n u_i dx_i$, $A\in O(1,n-1)$. So (\ref{E:flat-1system-1}) implies that the $1$-dimensional twisted $\frac {O(J,J)}{O(J)\times O(J)}$-system is equivalent to the flatness condition
\begin{equation}
d\theta_\lambda+\theta_\lambda\wedge \theta_\lambda=0.\label{E:flatness}
\end{equation}On the other hand, by (\ref{E:flat-1system-2}), the flatness condition (\ref{E:flatness}) is equivalent to $(A,u)$ satisfying the following system
\begin{equation*}
\left\{
  \begin{array}{l}
     -\delta \wedge dA^t + u\wedge \delta A^t =0 \Leftrightarrow \delta\wedge (JA^{-1}dAJ)+ u\wedge \delta=0, \quad {(*)}\\
     (JA^{-1}dAJ)\wedge \delta+ \delta\wedge u=0 \Leftrightarrow (*),\\
    du+u\wedge u+\delta A^t e_1\wedge JAJ\delta=0.
  \end{array}
\right.
\end{equation*}
The first equation implies that there exists $H=(h_{ij})$ with $h_{ii}=0$ for all $1\leq i\leq n$ such that
$u=\delta H-JH^t \delta J$, $JA^{-1}dA J=J\delta H^t J -H\delta$ $\Leftrightarrow$ $A^{-1}dA=\delta H^t-JH\delta J$.
Thus the $1$-dimensional twisted $\frac{O(J,J)}{O(J)\times O(J)}$-system is given by the following PDEs:
\begin{equation}\label{eq:twistedsystem}
\begin{split}
&u=\delta H-JH^t \delta J, \\
&A^{-1}dA=\delta H^t-JH\delta J, \\
&du+u\wedge u+\delta A^t e_1\wedge JAJ\delta=0.
\end{split}
\end{equation} Comparing (\ref{E:GC-1}) with (\ref{eq:twistedsystem}), the first assertion of Theorem \ref{T:GC-sinhG} is proved by setting $H=F$, $u=\omega$.

$\underline{\textsl{Step 3: the generalized sinh-Gordon equation}}$

Formula (\ref{E:GSHGE2}) is exactly (\ref{eq:a^k_i_j}). Taking the coefficients of $dx_i\wedge dx_k$ of the $ij$-entry of both sides of the first equation of (\ref{E:GC-1}), we obtain (\ref{E:GSHGE4}). Similarly, (\ref{E:GSHGE3}) can be derived by taking the coefficients of $dx_i\wedge dx_j$ of the $ij$-entry of both sides of the first equation of (\ref{E:GC-1}).  
\end{proof}

\begin{proof} $\underline{\textsl{of {\bf Example \ref{Ex:twistedmw}} }}$:

Theorem \ref{T:GC-sinhG} reduces the proof to showing the existence of such $n$-dimensional submanifolds.  Note Theorem \ref{T:GSHGE} of Section \ref{S:cauchy} implies that the $1$-dimensional twisted $\frac {O(J,J)}{O(J)\times O(J)}$-system (\ref{E:1Dsystem1}) with $b$, $a_i$, $v_i$ defined by (\ref{E:gav}) can be solved in $\mathbb R^{n}$. By Theorem \ref{T:GC-sinhG}, we then conclude the solvability of the Gauss-Codazzi equation of such submanifolds. Therefore, a modifed version of the Bonnet Theorem yields the  existence of a time-like $n$-dimensional submanifold $M$ of constant sectional curvature $1$ in $\mathbb R^{2n-1}_1$, local coordinates $x_1,x_2,\cdots, x_n$ on a neighborhood of $p\in M$,
and parallel normal frames $e_{n+1}, \cdots, e_{2n-1}$ with
the first and second fundamental forms  (\ref{E:I-II}).
\end{proof}

We remark that the correspondence between the sinh-Gordon equation and the positive constant Gaussian curvature time-like surface in $\mathbb R^3_1$ has been established by Chern \cite{C81}. In the following theorem, we construct a Riccati type B$\ddot {\rm a}$cklund transformation, analogous to (\ref{E:Backlund}), of the generalized sinh-Gordon equation (\ref{E:GSHGE1})-(\ref{E:GSHGE4}). Moreover, we linearize the B$\ddot {\rm a}$cklund transformation.
\begin{theorem}
Suppose $A\in O(1,n-1)$ is a solution of the generalized sinh-Gordon equation and $\lambda$ is a non-zero \textbf{real} constant. 
Consider the linear system for
$y: \mathbb{R}^n \rightarrow \mathcal{M}_{n\times 2n}$:
\begin{equation}\label{BT}
dy=y \left(
             \begin{array}{cc}
               \omega & \delta A^t D_\lambda \\
               D_\lambda JA J\delta & 0 \\
             \end{array}
           \right),
\quad D_\lambda=\frac{1}{2}(\lambda I-\lambda^{-1}J).
\end{equation}
Then
\begin{enumerate}
	\item  System (\ref{BT}) is solvable.
	\item  If $y=(P, Q)$ is a solution of \eqref{BT} with $Q\in GL(n)$, 
then $X=-Q^{-1}P\in O(1,n-1)$ is a solution of the B$\ddot {\rm a}$cklund transformation for the generalized sinh-Gordon equation given by 
\begin{equation}\label{E:Backlund-GSHGE}
dX=X\delta A^t D_\lambda X+X\omega -D_\lambda JAJ\delta.
\end{equation}
and $X$ is again a solution of the generalized sinh-Gordon equation.
\end{enumerate}
\end{theorem}

\begin{proof} Define $\theta_\lambda$ by
\[
\theta_\lambda=\left(
             \begin{array}{cc}
               \omega & \delta A^t D_\lambda \\
               D_\lambda JA J\delta & 0 \\
             \end{array}
           \right).
           \] The assumption that
 $A$ is a solution of the generalized sinh-Gordon equation
gives the flatness of $\theta(\lambda)$ for any $\lambda\in \mathbb{C}$,
which can imply the solvability of (\ref{BT}).

To prove the second statement, let
\[d(P, Q)=(P, Q)\theta_\lambda\]
or equivalently,
\[\left\{
  \begin{array}{ll}
    dP=P\omega +QD_\lambda JAJ\delta, \\
    dQ=P\delta A^t D_\lambda.
  \end{array}
\right.
\]By a direct computation, (\ref{E:Backlund-GSHGE}) is satisfied. 

On the other hand, the assumption of $A\in O(1,n-1)$ being a solution of the generalized sinh-Gordon equation implies (\ref{E:GC-1}). Hence 
\begin{equation*}
\begin{split}
X^{-1}dX &= \delta A^t D_\lambda X+\omega-X^{-1}D_\lambda JAJ\delta\\
&=\delta A^t D_\lambda X + (\delta F-JF^t \delta J)-JX^tD_\lambda A\delta J\\
&=\delta (A^t D_\lambda X +F)-J(X^t D_\lambda A +F^t)\delta J\\
&:=\delta \tilde{F}^t -J\tilde{F} \delta J,
\end{split}
\end{equation*}
where $\tilde{F}=X^t D_\lambda A+F^t$.
Moreover, let
\begin{equation*}
\begin{split}
\tilde{\omega}&:=\delta \tilde{F}-J\tilde{F}^t \delta J\\
&= \delta(X^t D_\lambda A+F^t)-J(A^t D_\lambda X +F)\delta J\\
&= \delta X^t D_\lambda A -JA^t D_\lambda X\delta J+A^{-1}dA.
\end{split}
\end{equation*}

Then $(\tilde{\omega}, X)$ satisfy the following system:
$$
\left\{
  \begin{array}{ll}
    d\tilde{\omega}+\tilde{\omega}\wedge \tilde{\omega}+\delta X^t D_\lambda^2 \wedge (JX\delta J)=0 \\
    \tilde{\omega}=\delta \tilde{F}-J\tilde{F}^t \delta J \\
    X^{-1}dX=\delta \tilde{F}^t-J \tilde{F} \delta J,
  \end{array}
\right.
$$
where $\tilde{F}=X^t D_\lambda A+F^t$.
It is equivalent to
$$\tilde\theta_\lambda=\left(
             \begin{array}{cc}
               \tilde{\omega} & \delta X^t D_\lambda \\
               D_\lambda JX J\delta & 0 \\
             \end{array}
           \right)
$$
is flat.
By the argument in the proof of Theorem \ref{T:GC-sinhG}, $X$ is also a solution of the generalized sinh-Gordon equation.
\end{proof}

\begin{remark}\label{Ex:twisteduk2} (\textbf{The generating equation}) The Gauss-Codazzi equation for an $n$-dimensional Riemannian submanifold of constant sectional curvature $K$ with flat normal bundle in a $(2n-1)$-dimensional Riemannian or pseudo-Riemannian manifold (of index $s$) of constant sectional curvature $\overline K$ is \textbf{\textit{the generating equation }}
\begin{eqnarray*}
	&\alpha\in O(n-q,q),&\\
	&\partial_{x_j}\alpha_{ki}=\alpha_{kj}f_{ji},\ f_{ii}=0, & i\ne j, \\
	&\partial_{x_i}f_{ij}+\partial_{x_j}f_{ji}+\sum_{k\ne i,\,j}f_{ki}f_{kj}=-K\alpha_{1i}\alpha_{1j},& i\ne j\\
	&\partial_{x_k}f_{ij}=f_{ik}f_{kj},& i,j,k\ \textsl{ distinct}
\end{eqnarray*}
where $1\le i,\,j,\,k\le n$, $q=s$ if $K<\overline K$ and $q=n-(s+1)$ if $K>\overline K$ \cite{CT94}, \cite{BFT96}, \cite{Ten98}. Moreover,  the B$\ddot {\rm a}$cklund transformation for the generating equation is constructed by showing that
\begin{eqnarray}
	&& dX+X\hat J^{-1/2}C\hat J^{1/2}=\Lambda_\lambda \alpha \delta \hat J^{1/2}-X\hat J^{-1/2}\delta \alpha^t\Lambda_\lambda  \hat JX	\label{E:Backlund-GE}
\end{eqnarray}
gives a new solution to the generating equation if $\alpha$ is a given soltuion of the generating equation. Here 
\begin{gather*}
\hat J=\mathrm{diag}(\overbrace{1,\cdots, 1}^{n-q \ \mathrm{ times}},\overbrace{-1,\dots,-1}^{q \ \mathrm{ times}}),\\
	 C=F\delta-\delta F^t,\qquad \delta=\sum_{j=1}^n e_jdx_j,\qquad \Lambda_\lambda =\frac 12(\lambda I+\frac K\lambda I_{1,n-1})
\end{gather*}
$F=(f_{ij})$, $e_j$ is defined as in (\ref{E:ai}),  and $I$ is the $n\times n$ identity matrix \cite{C94}, \cite{CT94}, \cite{Ten98}. Similarly, the B$\ddot {\rm a}$cklund transformation (\ref{E:Backlund-GE}) is linearized by the Lax pair
\begin{gather}
\partial_{x_j}\tilde \Psi = \left(\frac 12\lambda A_j\mp\frac 1{2\lambda}B_j+C_j\right)\tilde\Psi,\label{E:Lax-GE}
\end{gather}
(the $\mp$ sign corresponds to $K=\pm 1 $) with
\begin{eqnarray}
A_j &=& \left(
\begin{array}{cc}
0 & \alpha \hat J^{\frac 12}e_j\\
e_j \hat J^{-\frac 12}\alpha^t \hat J & 0
\end{array}
\right),\label{E:A-GE}\\
B_j &=& \left(
\begin{array}{cc}
0 & -I_{1,n-1}\alpha \hat J^{\frac 12} e_j\\
-e_j\hat J^{-\frac 12}\alpha^tI_{1,n-1}\hat J & 0
\end{array}
\right),\label{E:B-GE}\\ 
C_j &=& \left(
\begin{array}{lc}
0 & 0\\
0 & \hat J^{-\frac 12}\gamma_j \hat J^{\frac 12}
\end{array}
\right), \label{E:C-GE}
\end{eqnarray}
$\gamma_j\in o(n)$. However, it is impossible to transform (\ref{E:Lax-GE})-(\ref{E:C-GE}) 
 into a twisted $\frac {O(J,J)}{O(J)\times O(J)}$-system (twisted by $\sigma_1$)
because  the reality conditions fail by observing $A_j$, $B_j$, $C_j\notin o(J,J)$ unless $\hat J=I$. In particular, let $n=2$, 
\begin{eqnarray*}
&& \hat J=\left(\begin{array}{cc}1&0\\0&-1\end{array}\right),\\
&& \alpha=\left(\begin{array}{cc}\cosh\frac u2&\sinh\frac u2\\\sinh\frac u2&\cosh\frac u2\end{array}\right)\in O(1,1),\\ &&\gamma_1=\left(\begin{array}{cc}0&\alpha\\-\alpha &0\end{array}\right),\quad \gamma_2=\left(\begin{array}{cc}0&\beta\\-\beta &0\end{array}\right)\in o(n)
\end{eqnarray*}
in (\ref{E:A-GE})-(\ref{E:C-GE}), then the compatibility conditions of (\ref{E:Lax-GE}) are
\[\alpha=-\frac 12\partial_{x_2}u,\quad \beta=\frac 12\partial_{x_1}u,\quad \partial_{x_1}\beta-\partial_{x_2}\alpha=\mp\frac 12\sinh u.
\]Therefore
we obtain the \textit{{\bf{sinh-Laplace equation}}}
\begin{equation}
u_{x_1x_1}+u_{x_2x_2}=\mp \sinh u.\label{E:ellipticsinh}
\end{equation}
Here the $\mp$ sign corresponds to $K=\pm 1 $. We remark that the correspondence between the sinh-Laplace equation and the negative constant Gaussian curvature space-like surface in $\mathbb R^3_1$ has been discovered by Hu \cite{Hu85}.
\end{remark}

\section{The direct scattering problem}\label{S:invTH}
Using (\ref{E:pi}), the linear spectral problem corresponding to (\ref{E:j-thflow}) is
\begin{gather}
\frac{\partial \Psi}{\partial x}=-\lambda bab^{-1}\Psi-\frac 1\lambda\sigma_1(bab^{-1})\Psi-v\Psi, \label{E:lax}\\
b\in K'_0,\quad v\in\mathcal S_0.\nonumber
\end{gather}
In this section, we center on the construction of special eigenfunctions $\Psi(x,\lambda)$. By the normalization
\begin{eqnarray}
\Psi(x,\lambda)&=& m(x,\lambda)e^{-x(\lambda a +\frac 1\lambda\sigma_1(a))}\nonumber\\
&=& b \breve{m} (x,\lambda)e^{-x(\lambda a +\frac 1\lambda\sigma_1(a))},\label{E:tildem}
\end{eqnarray}
the linear spectral problem (\ref{E:lax}) turns into
\begin{eqnarray}
\frac{\partial m}{\partial x}&=&\lambda \left(ma-bab^{-1}m\right)+\frac 1\lambda\left(m\sigma_1(a)-\sigma_1(bab^{-1})m\right)-vm,\label{E:normalizedlax}
\\
\frac{\partial \breve{m}}{\partial x}&=&  [ \breve{m}(x,\lambda),\,\, \lambda a+\frac 1\lambda\sigma_1(a)]+Q(x,\lambda) \breve{m}(x,\lambda),\label{E:m'}
\end{eqnarray}
with
\begin{eqnarray}
Q(x,\lambda)
&=&\frac 1\lambda\left(\sigma_1(a)-b^{-1}\sigma_1(bab^{-1})b\right)- (b^{-1}\frac {\partial b}{\partial x}+b^{-1}vb ).\label{E:Q}
\end{eqnarray}
\begin{definition}\label{E:proj}We define the operator $\mathcal J_\lambda$ on $gl(n,\mathbb C)$ by
\[\mathcal J_\lambda f=\left[f, \lambda a+\frac 1\lambda\sigma_1(a)\right], 
\]
and $\pi^\lambda_0$, $\pi^\lambda_\pm$ to be the orthogonal projections of $gl(n,\mathbb C)$ to the $0$--, $\pm$--eigenspaces 
	of $Re\, \mathcal J_\lambda=\frac 12(\mathcal J_\lambda+(\mathcal J_\lambda)^*)$. Moreover, 
the characteristic curve of (\ref{E:lax}) is defined by
\[\Sigma_a=\left\{\lambda\in\mathbb C|\ \textit{the image of $\pi^\lambda_0$ is non-empty}\right\}.\]
\end{definition}

\begin{definition}\label{Ex:principal} We call $a=\left(\begin{array}{cc} 0 & D\\ D &0\end{array}\right)$ a principal oblique direction if 
$|w_1|>|w_\nu|$ for $1<\nu\le n$, and the $2n$ real numbers $\left\{\pm w_1,\cdots,\pm w_n\right\}$ are distinct with $D=\textsl{diag }(w_1,\cdots,w_n)$ \cite{ABT86}. One can verify that if $a$ is a principal oblique direction, then   $\Sigma_a=i\mathbb R\cup \mathbb S^1$. Let us label the components of $\mathbb C\backslash \Sigma_a$ as
\begin{eqnarray*}
	\Omega^+ &=&\left\{|\lambda|>1, \ Re(\lambda)> 0\right\},\\
	\Omega^- &=&\left\{|\lambda|> 1, \ Re(\lambda)< 0\right\},\\
	D^+ &=&\left\{|\lambda|< 1, \ Re(\lambda)< 0\right\},\\
	D^- &=&\left\{|\lambda|< 1, \ Re(\lambda)> 0\right\}.
\end{eqnarray*}

\end{definition}

Note permutaion matrices do not commute with $\sigma_1$. Hence it is natural to consider the following example.
\begin{definition}\label{Ex:oblique} We call $a=\left(\begin{array}{cc} 0 & D\\ D &0\end{array}\right)$ an oblique direction if 
the $2n$ real numbers $\left\{\pm w_1,\cdots,\pm w_n\right\}$ are distinct. Here $D=\textsl{diag }(w_1,\cdots,w_n)$. One can verify that if $a$ is an oblique direction, then   $\Sigma_a=i\mathbb R\cup \mathbb S^1\cup_{1\le \nu\le s}\left(\mathbb S^{r_v}\cup \mathbb S^{1/r_\nu}\right)$, where $r_\nu=r_\nu(w_1,w_\nu)\ne 1$, and $s$ is the number of $w_\nu$ such that $|w_\nu|>|w_1|$. Let us label the components of $\mathbb C\backslash \Sigma_a$ as
\begin{eqnarray*}
	\Omega^+_\nu &=&\left\{r_\nu<|\lambda|< r_{\nu+1}, \ Re(\lambda)> 0\right\},\\
	\Omega^-_\nu &=&\left\{r_\nu<|\lambda|< r_{\nu+1}, \ Re(\lambda)< 0\right\},\\
	D^+_\nu &=&\left\{1/r_{\nu+1}<|\lambda|< 1/r_{\nu}, \ Re(\lambda)< 0\right\},\\
	D^-_\nu &=&\left\{1/r_{\nu+1}<|\lambda|<1/r_{\nu}, \ Re(\lambda)> 0\right\},
\end{eqnarray*}
for $0\le \nu\le s$. Here we assume $r_0=1<r_1<\cdots<r_s<r_{s+1}=\infty$.
\end{definition}

Restricted to the case of $q=0$  and $a$ is a principal oblique direction, the direct problem is solved by \cite{ABT86}, \cite{BT88} after a diagonalization process of (\ref{E:lax}). 
\begin{theorem}\label{T:existence}
Let $a\in\mathcal A$, a constant oblique direction, $b(x)\in K'_0$, $v(x)\in\mathcal S_0$. If 
$ |b-1|_{L_1^1\cap L_\infty}+ |v|_{L_1}<\infty$,
then there exists a bounded set $Z\subset\mathbb C$, such that $Z\cap\left(\mathbb C\backslash \Sigma_a\right)$ is discrete in $\mathbb C\backslash \Sigma_a$ and for $\forall\lambda\in\mathbb C\backslash \Sigma_a$, there exists uniquely a solution $m(x,\lambda)$ of (\ref{E:normalizedlax}) satisfying:
	\begin{eqnarray}
		&&\textit{$m(\cdot,\lambda)$ is bounded, for each $\lambda\in\mathbb C\backslash (\Sigma_a\cup Z)$},\label{E:ub}\\
	&&\textit{$m(x,\lambda)\to 1$ as $x\to -\infty$, for each $\lambda\in\mathbb C\backslash (\Sigma_a\cup Z)$,}\label{E:infty-x}\\
  &&\textit{$m(x,\cdot)$ is meromorphic in $\mathbb C\backslash \Sigma_a$ with poles at $\lambda\in Z$},\label{E:pole}\\
	&&m(x,\lambda)\to b(x) \textit{ uniformly as $\lambda\to\infty$,}\label{E:infty}
\end{eqnarray}
and
\begin{eqnarray}
	&& m(x,\lambda)\in L^{\sigma_0}_-,\label{E:ujj}\\
	&& m(x,\lambda)=
	\sigma_1(m(x,1/\lambda)\,).\label{E:sigma2}
	\end{eqnarray}
\end{theorem}
\begin{proof}

$\underline{\textsl{Step 1: (Small data problem)}}$

In this step, we assume that 
\[|Q |_{L_1}<1 ,\quad\textsl{ for $\forall |\lambda|\ge 1$,} 
\]where $Q$ is defined by (\ref{E:Q}). Thus for $(x,\lambda)\in \mathbb R\times\Omega^\pm_\nu$ we can find $m'(x,\lambda)$ which satisfies the integral equation
\begin{eqnarray*}
	m'(x,\lambda) = 1&+&\int_{-\infty}^x e^{(x-y)\mathcal J_\lambda}(\pi^\lambda_0+\pi^\lambda_-)\left(Q(y,\lambda)m'(y,\lambda)\right)\,dy\\
	&-&\int^{\infty}_x e^{(x-y)\mathcal J_\lambda}\pi^\lambda_+\left(Q(y,\lambda)m'(y,\lambda)\right)\,dy.
\end{eqnarray*}
One can also verify that
\begin{eqnarray*}
	&&\textit{$m'(x,\lambda)$ satisfies (\ref{E:m'}), (\ref{E:ub}), (\ref{E:ujj}) for $(x,\lambda)\in \mathbb R\times\Omega^\pm_\nu$}, \\
	&&\textit{$m'(x,\lambda)\to 1$  as $x\to -\infty$ or $|\lambda|\to\infty$,} \\
	&&\textit{$m'(x,\lambda)$ is holomorphic in $  \lambda\in\Omega^\pm_\nu$, and has a continuous extension} \\
	&&\textit{to $\Sigma_a$ from $\Omega^\pm_\nu$}.\nonumber
\end{eqnarray*}
Define 
\begin{equation}
m(x,\lambda)=\begin{cases}
b(x)m'(x,\lambda), &\textit{if $(x,\lambda)\in \mathbb R\times\Omega^\pm_\nu$};\\
\sigma_1(b(x)m'(x,\frac 1\lambda)), &\textit{if $(x,\lambda)\in\mathbb R\times D^\pm_\nu $}.
\end{cases}\label{E:m''}
\end{equation}
Using the $\sigma_1$--symmetry and the unique solvability of (\ref{E:normalizedlax}), we then prove the theorem provided $|Q |_{L_1}<1$, i.e. when the potentials $(b,v)$ satisfy 
$
 |b-1|_{L_1^1\cap L_\infty}+ |v|_{L_1}<c<<1$.

$\underline{\textsl{Step 2: (Large data problem)}}$

We induce on the least integer $N\ge 0$ such that $|Q |_{L_1}<2^N$. Note that the eigenfunction of (\ref{E:m'}) corresponding to a translate of $Q$ is the translate (with respect to $x$) of the eigenfunction $m'$. Thus without loss of generality, we may assume that
\begin{eqnarray*}
&& Q_-=\begin{cases}Q,&x\le 0,\\
0, &x\ge 0,
\end{cases} \quad Q=Q_+ +Q_-,\quad
 |Q_\pm|_{ L_1}<2^N.
\end{eqnarray*}
The induction assumption implies that $Q_-$ has an eigenfunction $\eta(x,\lambda)$, $Q_+$ has an eigenfunction $\rho(x,\lambda)$ (proved by analogy) satisfying 
\begin{eqnarray}
	&&\textit{$\eta(x,\lambda)$  satisfies (\ref{E:m'}), (\ref{E:ub}), (\ref{E:ujj}) for $(x,\lambda)\in \mathbb R^-\times\Omega^\pm_\nu$},\label{E:1}\\
	&&\textit{$\rho(x,\lambda)$  satisfies (\ref{E:m'}), (\ref{E:ub}), (\ref{E:ujj}) for $(x,\lambda)\in \mathbb R^+\times\Omega^\pm_\nu$,}\label{E:1'}\\
	&&\textit{$\eta(x,\lambda)\to 1$ as $x\to -\infty$, $\rho(x,\lambda)\to 1$  as $x\to \infty$,}\label{E:2}\\
	&&\textit{$\eta(x,\lambda)$, $\rho(x,\lambda)$ are meromorphic in $  \lambda\in\Omega^\pm_\nu$ and tend to $1$ as $\lambda\to \infty$.}\label{E:3}
\end{eqnarray}
Let us define
\begin{equation}
S^\pm(\lambda)=\rho^{-1}(0,\lambda)\eta(0,\lambda)\quad \textit{ for $\lambda\in\Omega^\pm_\nu$}.\label{E:factorization}
\end{equation}
One can adapt the argument in \cite{FF63} to factorize 
\begin{equation}
S^\pm(\lambda)=(1+L^\pm(\lambda))\delta^\pm(\lambda)(1+U^\pm(\lambda))^{-1},\label{E:factorization-0}
\end{equation}
where
\begin{eqnarray}
&&\textit{$\pi^\lambda_0(U^\pm(\lambda))=\pi^\lambda_-(U^\pm(\lambda))=0$, }\label{E:factorization-1}\\ &&\textit{$\pi^\lambda_0(L^\pm(\lambda))=\pi^\lambda_+(L^\pm(\lambda))=0$},\label{E:factorization-2}\\
&&\textit{$\pi^\lambda_+(\delta^\pm(\lambda))=\pi^\lambda_-(\delta^\pm(\lambda))=0$},\label{E:factorization-3}
\end{eqnarray}
for $\lambda\in\Omega^\pm_\nu$ and
\begin{eqnarray}
&&\textit{$U^\pm$, $L^\pm$, $\delta^\pm$ are meromorphic in $\lambda\in\Omega^\pm_\nu$ with poles at $Z$, }	\label{E:factorization-4}\\
&&\textit{$U^\pm$, $L^\pm$ tend to $0$, $\delta^\pm(\lambda)$ tends to $1$ uniformly as $|\lambda|\to\infty$.}\label{E:factorization-5}
\end{eqnarray}
Where $Z=\{\textit{zeros of minors of $P_{\pm,\nu}^{-1}S^\pm P_{\pm,\nu}$}\}$ and $P_{\pm,\nu}$ satisfies that 
\begin{gather*}
\textit{$P_{\pm,\nu}^{-1} J_{a,1} P_{\pm,\nu}$ is diagonal},\\
\textit{real parts of the entries of $P_{\pm,\nu}^{-1} J_{a,1} P_{\pm,\nu}$ are nondecreasing} 
\end{gather*}
on $\Omega^\pm_\nu$. Define 
\begin{equation}
m'(x,\lambda)=\begin{cases}
\eta (x,\lambda) e^{x\mathcal J_\lambda}\left(1+U^\pm(\lambda)\right), &(x,\lambda)\in\{x\le 0\}\times\Omega^\pm_\nu;\\
\rho (x,\lambda)\left(e^{x\mathcal J_\lambda}(1+L^\pm(\lambda))\right)\delta^\pm(\lambda), &(x,\lambda)\in\{x\ge 0\}\times\Omega^\pm_\nu,
\end{cases}\label{E:extendm}
\end{equation}
and $m(x,\lambda)$ by (\ref{E:m''}). Then we complete the theorem by properties (\ref{E:m''})-(\ref{E:extendm}). 
\end{proof}


\begin{corollary}\label{C:existence-1}
Let $a\in\mathcal A$, a constant oblique direction, $b(x)\in K'_0$, $v(x)\in\mathcal S_0$, 
$ |b-1|_{L_1^1\cap L_\infty}+ |v|_{L_1}<\infty$. If the set $Z$ in Theorem \ref{T:existence} is a finite set contained in $\mathbb C\backslash\Sigma_a$, then we have the following factorization properties for ${m}(x,\lambda)$:
\begin{equation}
m(x,\lambda)=\begin{cases}
b(x)\eta^\pm_\nu (x,\lambda) e^{x\mathcal J_\lambda}\left(1+U^\pm_\nu(\lambda)\right), &(x,\lambda)\in\{x\le 0\}\times\Omega^\pm_\nu,\\
b(x)\rho^\pm_\nu (x,\lambda)e^{x\mathcal J_\lambda}(1+L^\pm_\nu(\lambda))\delta^\pm_\nu(\lambda), &(x,\lambda)\in\{x\ge 0\}\times\Omega^\pm_\nu;\\
b(x)\tilde \eta^\pm_\nu (x,\lambda) e^{x\mathcal J_\lambda}(1+\tilde U^\pm_\nu(\lambda)), &(x,\lambda)\in\{x\le 0\}\times D^\pm_\nu,\\
b(x)\tilde \rho^\pm_\nu (x,\lambda)e^{x\mathcal J_\lambda}(1+\tilde L^\pm_\nu(\lambda))\tilde \delta^\pm_\nu(\lambda), &(x,\lambda)\in\{x\ge 0\}\times D^\pm_\nu,
\end{cases}\label{E:factm}
\end{equation} with
\begin{eqnarray*}
&&\textit{$\eta^\pm_\nu(x,\lambda)$, $\rho^\pm_\nu(x,\lambda)$ satisfy (\ref{E:m'}), are uniformly bounded, and tend }\\
&&\textit{to $1$ as $x\to \mp\infty$ respectively,}\\
&&\textit{$\eta^\pm_\nu(x,\lambda)$, and $\rho^\pm_\nu(x,\lambda)$ are holomorphic and tend to $1$ as $|\lambda|\to\infty$, }\\
&&\textit{$\delta^\pm_\nu(\lambda)$ are meromorphic and tend to $1$ as $|\lambda|\to\infty$,}\\
&&{}\\
&&\textit{$\pi^\lambda_0(U^\pm_\nu(\lambda))=\pi^\lambda_-(U^\pm_\nu(\lambda))=0$, }\\ &&\textit{$\pi^\lambda_0(L^\pm_\nu(\lambda))=\pi^\lambda_+(L^\pm_\nu(\lambda))=0$},\\
&&\textit{$\pi^\lambda_+(\delta^\pm_\nu(\lambda))=\pi^\lambda_-(\delta^\pm_\nu(\lambda))=0$},
\end{eqnarray*}
for $\lambda\in\Omega^\pm_\nu$, and
\begin{eqnarray*}
&&\textit{$U^\pm_\nu$, $L^\pm_\nu$ are rational in $\lambda\in\Omega^\pm_\nu$, holomorphic in $\lambda\in\left(\Omega^\pm_\nu\right)^c$,}\\
&&\textit{$U^\pm_\nu$, $L^\pm_\nu$ tends to $0$ as $|\lambda|\to\infty$.}
\end{eqnarray*}
Besides,
\begin{eqnarray*}
b(x)\tilde \eta^\mp_\nu (x,\lambda) &=& \sigma_1(b(x)\eta^\pm_\nu (x, 1/\lambda)(1+U^\pm_\nu(0)),\\
b(x)\tilde \rho^\mp_\nu (x,\lambda) &=& \sigma_1(b(x)\rho^\pm_\nu (x,1/\lambda) (1+L^\pm_\nu(0))\delta^\pm_\nu(0)),\\
1+\tilde U^\mp_\nu(\lambda) &=& \sigma_1\left((1+U^\pm_\nu(0))^{-1}(1+U^\pm_\nu(1/\lambda))\right), \\
1+\tilde L^\mp_\nu(\lambda) &=& \sigma_1\left(\delta^\pm_\nu(0)^{-1}(1+L^\pm_\nu(0))^{-1}(1+L^\pm_\nu(1/\lambda))\delta^\pm_\nu(0)\right), \\
\tilde\delta^\mp_\nu(\lambda)&=&\sigma_1(\delta^\pm_\nu(0)^{-1}\delta^\pm_\nu(1/\lambda)).
\end{eqnarray*}
Finally, if we denote by $m_+$ the limits on $\Sigma_a$ from components $\Omega_0^+$, $\Omega_1^-$, $\Omega_2^+$, $\Omega_3^-$, $\cdots$ and from $D_0^+$, $D_1^-$, $D_2^+$, $D_3^-$, $\cdots$, and denote by  $m_-$ the limits from the other components, then
\begin{equation}
m_+(x,\lambda)=m_-(x,\lambda)e^{-x(\lambda a +\frac 1\lambda\sigma_1(a))}V(\lambda)e^{x(\lambda a +\frac 1\lambda\sigma_1(a))}\quad\textit{for $\lambda\in \Sigma_a$}.\label{E:jump}
\end{equation}
\end{corollary}
\begin{proof} A refined argument of the proof of Theorem \ref{T:existence} can derive the factorization of $m(x,\lambda)$ on $\mathbb R\times \Omega^\pm_\nu$. The jump condition (\ref{E:jump}) comes from the limits $\breve{m}_\pm$, defined by (\ref{E:tildem}), exist on $\Sigma_a$, $\breve{m}_\pm$ satisfy the same equation (\ref{E:m'}) and the operator $\partial_x-\mathcal J_\lambda$ is a derivation.
\end{proof}

\begin{definition}\label{D:scattering} 
If the assumption in Corollary \ref{C:existence-1} holds, 
then $(U^\pm_\nu, V)$  is called the associated scattering data of the potential $(b,\,v)$.
\end{definition}

\begin{definition}\label{D:diagonalization}
Let $P(\lambda)$ be the matrix satisfying that $P^{-1}\mathcal J_\lambda P$ is a diagonal matrix with decreasing entries and $P_\pm=\lim_{\lambda_n\to \lambda}P(\lambda_n)$, $\lambda_n\in \Omega^\pm_0\cup \Omega^\mp_1\cup \Omega^\pm_2\cup\cdots\cup D_0^\pm\cup D_1^\mp\cup\cdots$. 
\end{definition}
Note that $P$ is constant on each component of $\mathbb C/\Sigma_a$.

\begin{theorem}\label{T:sc}
Let $a\in\mathcal A$, a constant oblique direction. Suppose $b(x)\in K'_0$, $v(x)\in\mathcal S_0$, 
and their derivatives are rapidly decreasing as $|x|\to \infty$. If the set $Z$ in Theorem \ref{T:existence} is a finite set contained in $\mathbb C\backslash\Sigma_a$.  Then for the scattering data $(U^\pm_\nu, V)$, we have the analytical constraints
\begin{eqnarray}
&&\textit{$\partial_\lambda^\alpha (V-I)$ is $O(\lambda^N)$ as $\lambda\to 0$ and $O(\lambda^{-N})$ as $\lambda\to \infty$  for $N$, $\alpha \ge 0$,}\label{E:ana1}
\\
&&\textit{the product of limits of $V$ from each component, arranged clockwisely, } \label{E:ana2}\\
&&\textit{is $I$ at each intersection of $\Sigma_a$; }\nonumber\\
&&\textit{$U^\pm_\nu$ are rational in $\lambda\in\Omega^\pm_\nu$, holomorphic in $\lambda\in\left(\Omega^\pm_\nu\right)^c$,} \label{E:ana2-u-1}\\
&&\textit{$U^\pm_\nu$ tend to $0$ as $|\lambda|\to\infty$.} \label{E:ana2-u-2}
\end{eqnarray}
the algebraic constraints
\begin{gather}
\textit{$\pi^\lambda_0(U^\pm_\nu(\lambda))=\pi^\lambda_-(U^\pm_\nu(\lambda))=0$, }\label{E:u-reality-122}\\
\tilde J\left(1+U^\pm_\nu(\bar\lambda)\right)^*\tilde J(1+U^\pm_\nu(\lambda))=1,\quad \left(U^\pm_\nu(\bar\lambda)\right)^*=\left(U^\pm_\nu\right)^t(\lambda),  \label{E:u-reality-1}\\
\sigma_0(U^\pm_\nu(-\lambda))=U^\mp_\nu(\lambda),
\label{E:u-reality-12}
\end{gather}
for $\lambda\in \Omega^\pm_\nu$ and
\begin{gather}
d_k^+(P^{-1}_+VP_+)=1, \,\ d_k^-(P^{-1}_+VP_+)\ne 0,\label{E:detv}\\
\tilde JV(\bar\lambda)^*\tilde J\,V(\lambda)=1,\ \left(V(\bar\lambda)\right)^*=V^t(\lambda),
 \label{E:u-reality-3}\\
\sigma_0(V(-\lambda))\,V(\lambda)=1, \label{E:u-reality-4}\\
\sigma_1 (V(1/\lambda))\,V(\lambda)=1 \label{E:u-reality-5}
\end{gather}
for $\lambda\in\Sigma_a$, $\forall 1\le k\le 2n$. Here $d_k^\pm (f)$ denote the upper and lower $k\times k$ minors of $f$.
\end{theorem}
\begin{proof}
Properties (\ref{E:ana2-u-1})-(\ref{E:u-reality-122}) have been shown in Corollary \ref{C:existence-1}. The analytic constraints (\ref{E:ana1}), and (\ref{E:ana2}) can be deduced from the results of \cite{BC84}. Note (\ref{E:ujj}) implies 
\begin{eqnarray}
\tilde J m(x,\bar\lambda)^* \tilde Jm(x,\lambda)=1,\ m(x,\bar\lambda)^*=m^t(x,\lambda),\ \sigma_0( m(x,-\lambda))=m(x,\lambda).\label{E:realitym}
\end{eqnarray} 
Together with the  the uniqueness of the factorization of $m(x,\lambda)$ on $\mathbb R\times \Omega^\pm_\nu$, we derive the reality conditions (\ref{E:u-reality-1}), (\ref{E:u-reality-12}).  
Similarly, (\ref{E:jump}) implies
\begin{eqnarray*}
\tilde J m_+(x,\bar\lambda)^*\tilde J &=&\tilde J e^{x(\lambda a +\frac 1{\lambda}\sigma_1(a))}V(\bar\lambda)^* e^{-x(\lambda a +\frac 1{\lambda}\sigma_1(a))} m_-(x,\bar\lambda)^*\tilde J\quad\textit{}\\
m_+(x,\bar\lambda)^*&=& e^{x(\lambda a +\frac 1{\lambda}\sigma_1(a))}V(\bar\lambda)^* e^{-x(\lambda a +\frac 1{\lambda}\sigma_1(a))} m_-(x,\bar\lambda)^* 
\end{eqnarray*}for $\lambda\in \Sigma_a$. Hence
\begin{eqnarray*}
m_+(x,\lambda)^{-1} &=& e^{-x(\lambda a +\frac 1{\lambda}\sigma_1(a))}\tilde J V(\bar\lambda)^* \tilde J e^{x(\lambda a +\frac 1{\lambda}\sigma_1(a))} m_-(x,\lambda)^{-1}\\
m_+(x,\lambda)^t &=& e^{x(\lambda a +\frac 1{\lambda}\sigma_1(a))}V(\bar\lambda)^* e^{-x(\lambda a +\frac 1{\lambda}\sigma_1(a))} m_-(x,\lambda)^t 
\end{eqnarray*}for $\lambda\in \Sigma_a$ 
by (\ref{E:realitym}). So (\ref{E:u-reality-3}) is proved. On the other hand, applying $\sigma_0$ to both sides of (\ref{E:jump}) and using (\ref{E:realitym}), we  obtain
\[
m_-(x,-\lambda)=m_+(x,-\lambda))e^{x(\lambda a +\frac 1\lambda\sigma_1(a))}\sigma_0(V(\lambda))e^{-x(\lambda a +\frac 1\lambda\sigma_1(a))}\quad\textit{for $\lambda\in \Sigma_a$}.
\] Hence we justify (\ref{E:u-reality-4}). Finally (\ref{E:u-reality-5}) is proved by applying $\sigma_1$  to both sides of (\ref{E:jump}) and using (\ref{E:sigma2}) instead. 

To prove (\ref{E:detv}), we compute
\begin{eqnarray*}
&&d_k^+(P^{-1}_+VP_+)(\lambda)\\
=&&d_k^+(P^{-1}_+(m_-^{-1}m^+)(0,\lambda)P_+)\\
=&&d_k^+(\ \left(P^{-1}_+(1+U_-)^{-1} \eta_-(0,\lambda)^{-1} \eta_+(0,\lambda)(1+U_+)P_+\right)\ )\\
=&&1
\end{eqnarray*}
by (\ref{E:factm}), (\ref{E:u-reality-122}), $\lim_{x\to -\infty}\eta^\pm_\nu =\lim_{x\to -\infty}b = 1$, and $\tilde\eta^\pm_\nu(x,\lambda)$ satisfying (\ref{E:m'}) \cite{S90}. 
Here 
\[
U(\lambda)=
\begin{cases}
U^\pm_\nu  \,(\lambda), &{\lambda\in\Omega^\pm_\nu},\\
\tilde U^\pm_\nu \,(\lambda), &{\lambda\in D^\pm_\nu},
\end{cases}\quad
\eta(x,\lambda)=
\begin{cases}
\eta^\pm_\nu  \,(x,\lambda), &{\lambda\in\Omega^\pm_\nu},\\
\tilde\eta^\pm_\nu(x,\lambda), &{\lambda\in D^\pm_\nu},
\end{cases} 
\] and $f_\pm$ the limits on $\Sigma_a$ from $\Omega_0^\pm\cup \Omega_1^\mp\cup\Omega_2^\pm\cup\Omega_3^\mp\cup\cdots\cup D_0^\pm\cup D_1^\mp\cup D_2^\pm\cup D_3^\mp\cup \cdots$. $d_k^-(P^{-1}_+VP_+)\ne 0$ can be proved by analogy.
\end{proof}

\begin{remark}\label{R:sigmar-1}
For $i\in\{0,\,1,\,\cdots,\,n-1\}$,  replacing $\sigma_1$ by $\sigma_i$ (defined in Remark \ref{R:sigmar}), we can solve the associated direct problem  by analogy.
\end{remark}

\section{The inverse scattering problem I }\label{S:inv-LGF}
The goal of the inverse problem is to find the potential $(b, v)$ for a given scattering data $(U^\pm_\nu, V)$. Usually we try to reverse the process in the direct problem. However there exists technical difficulties to find  $m$ since the boundary value of $m$ is $b(x)$ as $\lambda\to\infty$. Hence it is impossible for us to pose the Riemann-Hilbert problem for $m$.

We will adopt the approach of \cite{S90} to construct a normalized eigenfunction $\breve{m}(x,\lambda)$ prescribing the given scattering data $(U^\pm_\nu, V)$ with boundary value $1$ at infinity in this section. Then we will try to find a gauge which transforms $\breve{m}$ into a solution $m(x,\lambda)$ of (\ref{E:normalizedlax}) in next section. We note in \cite{ABT86}, \cite{BT88}, they solve the inverse problem for the twisted $\frac {O(n,n)}{O(n)\times O(n)}$-system with $U^\pm$ admitting only simple poles.

\begin{theorem}\label{T:q0}
Let $q=0$ and $a\in\mathcal A$. Suppose $(U^\pm_\nu, V)$ satisfies (\ref{E:ana1})-(\ref{E:u-reality-5}). Then there exists uniquely an $\breve{m}(x,\lambda)\in L^{\sigma_0}$ satisfying 
\begin{equation}
\textit{$\partial_x^{k'}\left(\breve{m}(x,\lambda)-I\right)\ \in L_2^k(\Sigma_a)$ for $\forall k,\,k'$, and tends to $0$ uniformly as $x\to -\infty$,}\label{E:regularityqo}
\end{equation}
\begin{equation}
\breve{m}(x,\lambda)=\begin{cases}
\eta^\pm_\nu (x,\lambda) e^{x\mathcal J_\lambda}\left(1+U^\pm_\nu(\lambda)\right), &(x,\lambda)\in\{x\le 0\}\times\Omega^\pm_\nu,\quad (1)\\
\rho^\pm_\nu (x,\lambda)e^{x\mathcal J_\lambda}(1+L^\pm_\nu(\lambda))\delta^\pm_\nu(\lambda), &(x,\lambda)\in\{x\ge 0\}\times\Omega^\pm_\nu,\quad (2)\\
\tilde \eta^\pm_\nu (x,\lambda) e^{x\mathcal J_\lambda}(1+\tilde U^\pm_\nu(\lambda)), &(x,\lambda)\in\{x\le 0\}\times D^\pm_\nu,\quad (3)\\
\tilde \rho^\pm_\nu (x,\lambda)e^{x\mathcal J_\lambda}(1+\tilde L^\pm_\nu(\lambda))\tilde \delta^\pm_\nu(\lambda), &(x,\lambda)\in\{x\ge 0\}\times D^\pm_\nu,\quad (4)
\end{cases}\label{E:brevem}
\end{equation}
and
\begin{equation}
\breve{m}_+(x,\lambda)= \breve{m}_-(x,\lambda)e^{-x(\lambda a +\frac 1\lambda\sigma_1(a))}V(\lambda)e^{x(\lambda a +\frac 1\lambda\sigma_1(a))}\quad\textit{for $\lambda\in \Sigma_a$}.\label{E:tildecs}
\end{equation}
Here 
\begin{eqnarray}
&&\textit{$\eta^\pm_\nu(x,\lambda)$, and $\rho^\pm_\nu(x,\lambda)$  are holomorphic and uniformly bounded,}\label{E:inv1}\\
&&\textit{$\eta^\pm_\nu(x,\lambda)$, and $\rho^\pm_\nu(x,\lambda)$ tend to $1$ as $|\lambda|\to\infty$, }\label{E:inv2}\\
&&\textit{$\delta^\pm_\nu(\lambda)$, $\tilde\delta^\pm_\nu(\lambda)$ are meromorphic, and tend to $1$ as $|\lambda|\to\infty$,}\label{E:inv3}\\
&&\textit{$\pi^\lambda_0(L^\pm_\nu(\lambda))=\pi^\lambda_+(L^\pm_\nu(\lambda))=0$},\label{E:inv4}\\
&&\textit{$\pi^\lambda_+(\delta^\pm_\nu(\lambda))=\pi^\lambda_-(\delta^\pm_\nu(\lambda))=0$},\label{E:inv5}
\end{eqnarray}
for  $\lambda\in\Omega^\pm_\nu$, 
\begin{eqnarray}
&&\textit{$\tilde\eta^\pm_\nu(x,\lambda)$, and $\tilde\rho^\pm_\nu(x,\lambda)$  are holomorphic and uniformly bounded}\label{E:inv6}
\end{eqnarray}
for  $\lambda\in D^\pm_\nu$, and
\begin{eqnarray}
&&\textit{$L^\pm_\nu$ are rational in $\lambda\in\Omega^\pm_\nu$, holomorphic in $\lambda\in\left(\Omega^\pm_\nu\right)^c$,}\label{E:inv7}\\
&&\textit{$L^\pm_\nu$ tends to $0$ as $|\lambda|\to\infty$,}\label{E:inv8}\\
&&1+\tilde U^\mp_\nu(\lambda) = \sigma_1\left((1+U^\pm_\nu(0))^{-1}(1+U^\pm_\nu(1/\lambda))\right),\label{E:inv9}\\
&&1+\tilde L^\mp_\nu(\lambda) = \sigma_1\left(\delta^\pm_\nu(0)^{-1}(1+L^\pm_\nu(0))^{-1}(1+L^\pm_\nu(1/\lambda))\delta^\pm_\nu(0)\right), \label{E:inv10}\\
&&\tilde\delta^\mp_\nu(\lambda)=\sigma_1(\delta^\pm_\nu(0)^{-1}\delta^\pm_\nu(1/\lambda)).\label{E:inv11}
\end{eqnarray}
\end{theorem}
\begin{proof} 

$\underline{\textsl{Step 1: $ \{x\le 0\}\times \left(\mathbb C\backslash \Sigma_a\right)$}}$

Define $\tilde U^\pm_\nu$ by (\ref{E:inv9}).  Let  
\begin{equation}
U(\lambda)=
\begin{cases}
U^\pm_\nu  \,(\lambda), &{\lambda\in\Omega^\pm_\nu},\\
\tilde U^\pm_\nu \,(\lambda), &{\lambda\in D^\pm_\nu},
\end{cases}\label{E:U}
\end{equation}
and $U_\pm$ the limits on $\Sigma_a$ from $\Omega_0^\pm\cup \Omega_1^\mp\cup\Omega_2^\pm\cup\Omega_3^\mp\cup\cdots\cup D_0^\pm\cup D_1^\mp\cup D_2^\pm\cup D_3^\mp\cup \cdots$. Define
\begin{equation}
W(\lambda)=(1+U_-(\lambda))V(\lambda)(1+U_+(\lambda))^{-1}\quad\textit{for $\lambda\in\Sigma_a$}.\label{E:W4}
\end{equation}  
Then to show  (1), (3) in (\ref{E:brevem}), and the statement about $\eta^\pm_\nu$ 
in (\ref{E:inv1}), (\ref{E:inv2}) is equivalent to solving the following Riemann-Hilbert problem with purely continuous scattering data
\begin{gather}
f_+(x,\lambda)=f_-(x,\lambda)e^{-x(\lambda a +\frac 1\lambda\sigma_1(a))}W(\lambda)e^{x(\lambda a +\frac 1\lambda\sigma_1(a))}\quad\textit{for $\lambda\in \Sigma_a$,}\label{E:rh}\\
\textit{$f$ is holomorphic in $\mathbb C/\Sigma_a$, $f(x,\lambda)\to 1$, as $\lambda\to\infty$}\label{E:rh-b}
\end{gather} and setting 
\[
f(x,\lambda)=
\begin{cases}
\eta^\pm_\nu(x,\lambda), &\lambda\in \Omega^\pm_\nu,\\
\tilde\eta^\pm_\nu(x,\lambda), &\lambda\in D^\pm_\nu.
\end{cases}
\]

To solve the above Riemann-Hilbert problem, one 
can apply the method (\S 10 in \cite{BC84}) to prove its Fredholm property. So the solvability is reduced to showing that the homogeneous solution is trivial. Let $g(x,\lambda)=f(x,\lambda)\,f(x,-\bar\lambda)^*$ and $f$ satisfies (\ref{E:rh}), is holomorphic in $\mathbb C/\Sigma_a$, and tends to $ 0$ as $\lambda\to\infty$. Hence   if 
\begin{equation}
W(-\bar\lambda)^*=W(\lambda),\quad\lambda\in\Sigma_a,\label{E:realityinv}
\end{equation}
then for $\lambda\in \Sigma_a$, 
\begin{eqnarray*}
g_-(x,\lambda)&=&f_-(x,\lambda)\,\left(f_+(x,-\bar\lambda)\right)^*\\
&=&f_-(x,\lambda)\, \left(f_-(x,-\bar\lambda)e^{x(\bar\lambda a+\frac 1 {\bar\lambda} \sigma_1(a))}W(-\bar\lambda)e^{-x(\bar\lambda a+\frac 1 {\bar\lambda} \sigma_1(a))}\right)^*\\
&=&f_-(x,\lambda)\,\left(e^{x\mathcal J}W(\lambda)\right)\,f_-(x,-\bar\lambda)^*\\
&=&f_+(x,\lambda)f_-(x,-\bar\lambda)^*\\
&=&g_+(x,\lambda).
\end{eqnarray*}
So $g=0$ by the Liouville's theorem. Moreover, the positivity condition (\ref{E:realityinv}) implies $f= 0$ for $\lambda\in i\mathbb R$. Hence $f\equiv 0$ by holomorphicy.

To prove the claim (\ref{E:realityinv}), we note that the algebraic constraints (\ref{E:u-reality-1}), (\ref{E:u-reality-12}), (\ref{E:u-reality-3}) and (\ref{E:u-reality-4}) of the scattering data imply
\begin{equation}
V(-\bar\lambda)^*=V(\lambda),\qquad (1+U_\pm(-\bar\lambda)\,)^*=(1+U_\mp(\lambda)\,)^{-1}.\label{E:invstep1}
\end{equation}
Using these reality conditions and the definition of $W$, we can derive (\ref{E:realityinv}). 

$\underline{\textsl{Step 2: Construction of $ \delta^\pm_\nu$, $\tilde\delta^\pm_\nu$, $L^\pm_\nu$, $\tilde L^\pm_\nu$}}$

We first solve the scalar Riemann-Hilbert problems for the entries of a matrix $\Delta(\lambda)$:
\begin{eqnarray}
&&\Delta \textit{ is meromorphic in $\mathbb C/\Sigma_a$, $\Delta(\lambda)\to 1 $ as $\lambda\to\infty$,}\label{E:deltainv-1}\\
&&\pi^\lambda_+(\Delta(\lambda))=\pi^\lambda_-(\Delta(\lambda))=0, \label{E:deltainv-0}\\
&&\left[P_+^{-1}\Delta_+P_+\right]_{k}=\left[P_+^{-1}\Delta_- P_+\right]_{k}\frac{d_{2n-k+1}^-\left(P_+^{-1}V P_+\right)}{d_{2n-k}^-\left(P_+^{-1}V P_+\right)},\label{E:deltainv-2}\\
&&\left\{\textit{ zeros of $\left[P^{-1}\Delta P\right]_k$ } \right\}=\left\{\textit{ poles of the $k+1$-th column of $P^{-1}UP$}\right\},\label{E:deltainv-3}\\
&&\left\{\textit{ poles of $\left[P^{-1}\Delta P\right]_k$ } \right\}=\left\{\textit{ poles of the $k$-th column of $P^{-1}UP$}\right\}.\label{E:deltainv-4}
\end{eqnarray} 
Here  $P$ is defined by Definition \ref{D:diagonalization}, $[f]_k$ denotes the $k$-th entry of the diagonal part of $f$, for $1\le k\le 2n$. 
Then $ \delta^\pm_\nu$, and $\tilde\delta^\pm_\nu$ are constructed by
\begin{equation}
\Delta(\lambda)=
\begin{cases}
\delta^\pm_\nu  \,(\lambda), &{\lambda\in\Omega^\pm_\nu},\\
\tilde \delta^\pm_\nu \,(\lambda), &{\lambda\in D^\pm_\nu}.
\end{cases}\label{E:defdelta}
\end{equation}

To construct $L^\pm_\nu$, $\tilde L^\pm_\nu$, by (\ref{E:detv}), (\ref{E:deltainv-0}) and (\ref{E:deltainv-2}), one can factorize 
\begin{eqnarray}
&&V(\lambda)=(1+\hat L^-(\lambda))^{-1}\Delta_-^{-1} (\lambda)\Delta_+ (\lambda)(1+\hat L^+(\lambda)),\label{E:dec1}\\
&&\pi ^{\lambda^+_n}_+(\hat L^+)=\pi ^{\lambda^+_n}_0(\hat L^+)=0,\nonumber\\
&&\pi ^{\lambda^-_n}_+(\hat L^-)=\pi ^{\lambda^-_n}_0(\hat L^-)=0,\nonumber
\end{eqnarray}
on $\Sigma_a$. Here $\lambda^\pm_n\to\lambda$, $\lambda^\pm_n\in\Omega_0^\pm\cup \Omega_1^\mp\cup\Omega_2^\pm\cup\Omega_3^\mp\cup\cdots\cup D_0^\pm\cup D_1^\mp\cup D_2^\mp\cup D_3^\mp\cup \cdots$. By Theorem 5.2 of \cite{S90}, we can have the extensions $L^{\pm,0}$, $R^{\pm,0}$ satisfying
\begin{eqnarray}
&&\Delta_\pm(\lambda)(1+\hat L^\pm(\lambda))\Delta_\pm^{-1}(\lambda)=(1+R^{\pm}(\lambda))(1+L^{\pm}(\lambda))\textit{ on $\Sigma_a$},\label{E:dec2}\\
&&\pi ^{\lambda^\pm_n}_+(R^{\pm})=\pi ^{\lambda_n^\pm}_0(R^{\pm})=0,\ \pi ^{\lambda_n^\pm}_+(L^{\pm})=\pi ^{\lambda_n^\pm}_0(L^{\pm})=0\textsl{ on $\mathcal O^\pm$},\label{E:dec2-2}\\
&&\textit{$R^{\pm}$ is holomorphic on $\mathcal O^\pm$, meromorphic on $(\mathcal O^\pm)^c$},\label{E:dec2-3}\\
&&\textit{$L^{\pm}$ is holomorphic on $(\mathcal O^\pm)^c$, rational on $\mathcal O^\pm$}.\label{E:dec2-4}
\end{eqnarray}
Here $\mathcal O^\pm$ is the component of $\Omega_0^\pm\cup \Omega_1^\mp\cup\Omega_2^\pm\cup\Omega_3^\mp\cup\cdots\cup D_0^\pm\cup D_1^\mp\cup D_2^\pm\cup D_3^\mp\cup \cdots$ which contains $\lambda^\pm_n$, $\lambda^\pm_n\to \lambda$. Hence (\ref{E:dec1}), (\ref{E:dec2}) imply
\begin{eqnarray}
V(\lambda)
&=& \Delta_-^{-1}(\lambda)(1+L_-(\lambda))^{-1}R(\lambda)(1+L_+(\lambda))\Delta_+(\lambda).\label{E:vr}
\end{eqnarray}
Where $\Delta$ is defined by (\ref{E:defdelta}), 
\begin{gather}
R=\left(1+ R^{-}(\lambda)\right)^{-1}(1+R^{+}(\lambda))\nonumber\\
L(\lambda)=
\begin{cases}
L^{+}(\lambda), &{\lambda\in\mathcal O^+},\\
L^{-}(\lambda), &{\lambda\in \mathcal O^-}.
\end{cases}\label{E:delL}
\end{gather}
Finally, $L^\pm_\nu(\lambda)$, $\tilde L^\pm_\nu(\lambda)$ are defined by 
\begin{equation}
L(\lambda)=
\begin{cases}
L^\pm_\nu  \,(\lambda), &{\lambda\in\Omega^\pm_\nu},\\
\tilde L^\pm_\nu \,(\lambda), &{\lambda\in D^\pm_\nu}.
\end{cases}\label{E:LV}
\end{equation}

$\underline{\textsl{Step 3: $ \{x\ge 0\}\times \left(\mathbb C\backslash \Sigma_a\right)$}}$

Having contructed $ \delta^\pm_\nu$, $\tilde\delta^\pm_\nu$, $L^\pm_\nu$, $\tilde L^\pm_\nu$, one can adopt the scheme of \textsl{Step 1} to find  $\rho^\pm_\nu$ and $\tilde\rho^\pm_\nu$. More precisely, let
\[\Lambda(\lambda)=(1+L_-(\lambda))\Delta_-(\lambda)V(\lambda)\Delta_+^{-1}(\lambda)(1+L_+(\lambda))^{-1}\quad\textit{for $\lambda\in\Sigma_a$}.
\]  
Then to show  (2), (4) in (\ref{E:brevem}), and the statement about $\rho^\pm_\nu$ 
in (\ref{E:inv1}), (\ref{E:inv2}) is equivalent to solving the following Riemann-Hilbert problem with purely continuous scattering data
\begin{gather}
f_+(x,\lambda)=f_-(x,\lambda)e^{-x(\lambda a +\frac 1\lambda\sigma_1(a))}\Lambda(\lambda)e^{x(\lambda a +\frac 1\lambda\sigma_1(a))}\quad\textit{for $\lambda\in \Sigma_a$,}\label{E:rh-r}\\
\textit{$f$ is holomorphic in $\mathbb C/\Sigma_a$, $f(x,\lambda)\to 1$, as $\lambda\to\infty$}\label{E:rh-b-r}
\end{gather} and setting 
\[
f(x,\lambda)=
\begin{cases}
\rho^\pm_\nu(x,\lambda), &\lambda\in \Omega^\pm_\nu,\\
\tilde\rho^\pm_\nu(x,\lambda), &\lambda\in D^\pm_\nu.
\end{cases}
\]
As above, the solvability can be reduced to showing:
\begin{equation}
\Lambda(-\bar\lambda)^*=\Lambda(\lambda),\quad\lambda\in\Sigma_a,\label{E:realityinv-r}
\end{equation}
and it can be implied by the reality conditions
\begin{equation}
V(-\bar\lambda)^*=V(\lambda),\ \Delta_\pm(-\bar\lambda)^*=\Delta_\mp(\lambda)^{-1}, \ (1+L_\pm(-\bar\lambda)\,)^*=(1+L_\mp(\lambda)\,)^{-1}.\label{E:invstep1-r}
\end{equation}
The reality condition $V(-\bar\lambda)^*=V(\lambda)$ follows from the algebraic constraints (\ref{E:u-reality-3}), (\ref{E:u-reality-4}). 
Moreover, by $V(-\bar\lambda)^*=V(\lambda)$, (\ref{E:dec1}), we have
\[
V(\lambda)=(1+\hat L^+(-\bar\lambda)^*)\Delta_+^* (-\bar\lambda)(\Delta_-^* (-\bar \lambda))^{-1}((1+\hat L^-(-\bar\lambda)^*)^{-1}
\]
Together with (\ref{E:detv}), we conclude $\left(\Delta(-\bar\lambda)^*\right)^{-1}$ satisfies (\ref{E:deltainv-2}). So there is no jump across $\Sigma_a$ for $\Delta(-\bar\lambda)^*\Delta(\lambda)$. On the other hand, by (\ref{E:deltainv-3}), (\ref{E:u-reality-1}), and (\ref{E:u-reality-12}),
\begin{eqnarray*}
&&\left\{\textit{ zeros of $\left[P^{-1}(\lambda)\left(\Delta (-\bar\lambda)^*\right)^{-1}P(\lambda)\right]_k$ } \right\}\\
=&&\left\{\textit{ zeros of $\left[P^{-1}(\lambda)(P(-\bar\lambda)^*)^{-1}P(-\bar\lambda)^*(\Delta (-\bar\lambda)^*)^{-1}\left(P(-\bar\lambda)^*\right)^{-1}P(-\bar\lambda)^*P(\lambda)\right]_k$ } \right\}\\
=&&\left\{\textit{ zeros of $\left[P(-\bar\lambda)^*(\Delta (-\bar\lambda)^*)^{-1}\left(P(-\bar\lambda)^*\right)^{-1}\right]_{2n-k+1}$ } \right\}\\
=&&\left\{\textit{ poles of the $2n-k$-th column of $P^*(-\bar\lambda)U(\lambda)P(-\bar\lambda)$}\right\}\\
=&&\left\{\textit{ poles of the $k+1$-th column of $P(\lambda)^{-1}(\lambda)U(\lambda)P(\lambda)$}\right\}\\
=&&\left\{\textit{ zeros of $\left[P^{-1}(\lambda)\Delta(\lambda)P(\lambda)\right]_k$ } \right\}
\end{eqnarray*}
Similarly,  (\ref{E:deltainv-4}),  (\ref{E:u-reality-1}), and (\ref{E:u-reality-12}) imply 
\begin{eqnarray*}
&&\left\{\textit{ poles of $ [P^{-1}(\lambda)\left(\Delta (-\bar\lambda)^*\right)^{-1} P(\lambda) ]_k$ } \right\}=\left\{\textit{ poles of $\left[P^{-1}(\lambda)\Delta(\lambda)P(\lambda)\right]_k$ } \right\}
\end{eqnarray*}
Thus the Liouville's theorem implies  $\Delta(-\bar\lambda)^*=\Delta(\lambda)^{-1}$ and $(1+\hat L^\pm(-\bar \lambda))^*=(1+\hat L^\mp(\lambda)^{-1}$. 
Finally
\begin{equation}
(1+   L^{\pm,0}(-\bar\lambda)\,)^*=(1+ L^{\mp,0}(\lambda)\,)^{-1}\label{E:lpm0}
\end{equation} by $V(-\bar\lambda)^*=V(\lambda)$, (\ref{E:dec1}), (\ref{E:dec2}), (\ref{E:dec2-3}), (\ref{E:dec2-4}), and the Liouville's theorem. So $(1+L_\pm(-\bar\lambda)\,)^*=(1+L_\mp(\lambda)\,)^{-1}$ comes from (\ref{E:delL}), (\ref{E:lpm0}), and $\Delta_\pm(-\bar\lambda)\,)^*=\Delta_\mp(\lambda)\,)^{-1}$.

$\underline{\textsl{Step 4: Proof of $\breve{m}(x,\lambda)\in L^{\sigma_0}$}}$

Let $W(\lambda)$ be defined by (\ref{E:W4}), then 
\begin{gather}
\tilde JW(\bar\lambda)^*\tilde J\,W(\lambda)=1,\ \left(W(\bar\lambda)\right)^*=W^t(\lambda),
 \label{E:u-reality-3-w}\\
\sigma_0(W(-\lambda))\,W(\lambda)=1, \label{E:u-reality-4-w}
\end{gather}
by (\ref{E:u-reality-1}), (\ref{E:u-reality-12}), (\ref{E:u-reality-3}), and (\ref{E:u-reality-4}). Therefore, one can reverse the process in the proof of Theorem \ref{T:sc} to show 
\begin{eqnarray}
\tilde J \breve{m}(x,\bar\lambda)^* \tilde J\breve{m}(x,\lambda)=1,\ \breve{m}(x,\bar\lambda)^*=\breve{m}^t(x,\lambda),\ \sigma_0( \breve{m}(x,-\lambda))=\breve{m}(x,\lambda).\label{E:realitym-inv}
\end{eqnarray}
So $\breve{m}(x,\lambda)\in L^{\sigma_0}$ for $x\le 0$. 

As for $x\ge 0$, note the formula (1)-(4), (\ref{E:realitym-inv}), and the unique factorization properties imply
\begin{gather*}
\tilde J\left(1+L^+_\nu(\bar\lambda)\right)^*\tilde J(1+L^+_\nu(\lambda))=1,\quad \left(L^\pm_\nu(\bar\lambda)\right)^*=\left(L^\pm_\nu\right)^t(\lambda),  \\
\sigma_0(L^-_\nu(-\lambda))=L^+_\nu(\lambda),
\\
\tilde J\delta^+_\nu(\bar\lambda)^*\tilde J\delta^+_\nu(\lambda)=1,\quad \delta^\pm_\nu(\bar\lambda)^*=(\delta^\pm_\nu)^t(\lambda),  \\
\sigma_0(\delta^-_\nu(-\lambda))=\delta^+_\nu(\lambda).
\end{gather*}
Therefore, using the above argument, one can justify $\breve{m}(x,\lambda)\in L^{\sigma_0}$ for $x\ge 0$.

$\underline{\textsl{Step 5: Proof of (\ref{E:inv10}) and (\ref{E:inv11})}}$

The strategy of the proof is analogous to that of showing (\ref{E:invstep1-r}).  By (\ref{E:u-reality-3}), (\ref{E:u-reality-5}), (\ref{E:dec1}), we have
\begin{equation}
V(\lambda)=(1+\sigma_1(\hat L^+(\frac 1\lambda))\,)^{-1}\sigma_1(\Delta_+ (\frac 1\lambda))^{-1}\sigma_1((\Delta_- (\frac 1\lambda))(1+\sigma_1(\hat L^-(\frac 1\lambda))\,)^{-1}\label{E:decV}
\end{equation}
Thus $\sigma(\Delta(\frac 1\lambda))$ satisfies (\ref{E:deltainv-2}) by (\ref{E:detv}) and $\sigma(\Delta(\frac 1\lambda))^{-1}\Delta(\lambda)$ has no jump across $\Sigma_a$. On the other hand, by (\ref{E:deltainv-3}), (\ref{E:inv9}), for $\lambda\in\mathbb C\backslash\Sigma_a$, 
\begin{eqnarray*}
&&\left\{\textit{ zeros of $\left[P(\lambda)^{-1}\sigma_1(\Delta (\frac 1{\lambda}))P(\lambda)\right]_k$ } \right\}\\
=&&\left\{\textit{ zeros of $\left[P(\lambda)^{-1}\sigma_1(P(\frac 1{\lambda}) P(\frac 1{\lambda})^{-1}\Delta (\frac 1{\lambda})
P(\frac 1{\lambda})P(\frac 1{\lambda})^{-1})P(\lambda)\right]_k$ } \right\}\\
=&&\left\{\textit{ zeros of $ \left[P(\frac 1{\lambda})^{-1}\Delta (\frac 1{\lambda})
P(\frac 1{\lambda})\right]_{k}$ } \right\}\\
=&&\left\{\textit{ poles of the $k+1$-th column of $P(\frac 1{\lambda})^{-1}U (\frac 1{\lambda})
P(\frac 1{\lambda})$}\right\}\\
=&&\left\{\textit{ poles of the $k+1$-th column of $P(\frac 1{\lambda})^{-1}\sigma_1(U (\lambda))
P(\frac 1{\lambda})$}\right\}\\
=&&\left\{\textit{ poles of the $k+1$-th column of $P({\lambda})^{-1}U (\lambda)
P({\lambda})$}\right\}\\
=&&\left\{\textit{ zeros of $\left[P^{-1}(\lambda)\Delta(\lambda)P(\lambda)\right]_k$ } \right\}
\end{eqnarray*}
Here we have use the diagonal property of $P(\lambda)^{-1}\sigma_1(P(\frac 1\lambda))$. Similarly,  
\begin{eqnarray*}
&&\left\{\textit{ poles of $ [P^{-1}(\lambda)\sigma_1(\Delta (\frac 1{\lambda})) P(\lambda) ]_k$ } \right\}=\left\{\textit{ poles of $\left[P^{-1}(\lambda)\Delta(\lambda)P(\lambda)\right]_k$ } \right\}
\end{eqnarray*}
can be justified by (\ref{E:deltainv-4}), (\ref{E:inv9}). Thus the Liouville's theorem implies  $\sigma_1(\Delta (\frac 1{\lambda}))^{-1}\Delta(\lambda)$ is a constant. Hence we prove (\ref{E:inv11}) by a normalization. Furthermore, $\sigma_1(\hat L^\pm(\frac 1\lambda))$ are equal to $\hat L^\mp(\lambda)$ up to a constant by (\ref{E:decV}). So (\ref{E:inv10}) follows by (\ref{E:dec2}), (\ref{E:dec2-3}),  (\ref{E:dec2-4}), (\ref{E:delL}), and (\ref{E:inv11}).

The analytic constraints (\ref{E:regularityqo}) can be deduced from the results of \cite{BC84}. 

\end{proof}

\begin{theorem}\label{T:qne0}
Let $q\ne 0$ and $a\in\mathcal A$. Suppose $V$ satisfies (\ref{E:ana1}), (\ref{E:ana2}), (\ref{E:detv})-(\ref{E:u-reality-5}) and $|V-1|_{L_\infty}<<1$. Then there exists uniquely an $\breve{m}(x,\lambda)\in L^{\sigma_0}$ satisfying 
\begin{eqnarray}
&&\breve{m}_+(x,\lambda)= \breve{m}_-(x,\lambda)e^{-x(\lambda a +\frac 1\lambda\sigma_1(a))}V(\lambda)e^{x(\lambda a +\frac 1\lambda\sigma_1(a))}\quad\textit{for $\lambda\in \Sigma_a$},\label{E:tildecsq0}\\
&&\textit{$\breve{m}$ is holomorphic in $\mathbb C/\Sigma_a$, $\breve{m}(x,\lambda)\to 1$, as $\lambda\to\infty$,}\label{E:tildecsq1}\\
&&\textit{$\partial_x^{k'}\left(\breve{m}(x,\lambda)-I\right)\ \in L_2^k(\Sigma_a)$ for $\forall k,\,k'$, and tends to $0$ uniformly }\label{E:regularityqneo}\\
&&\textit{as $x\to -\infty$.}\nonumber
\end{eqnarray}

\end{theorem}
\begin{proof}
If $q\ne 0$, then the positivity conditions (\ref{E:realityinv}), (\ref{E:realityinv-r}) fail in $\textsl{Step 1}$ and $\textsl{Step 3}$. However, the Riemann Hilbert problems (\ref{E:rh})-(\ref{E:rh-b}), (\ref{E:deltainv-1})-(\ref{E:deltainv-2}), and (\ref{E:rh-r})-(\ref{E:rh-b-r}), with $U(\lambda)\equiv 0$, $L(\lambda)\equiv 0$, in $\textsl{Step 1-3}$  can be solved under the small data constraint $|V-1|_{L^2}<<1$. 
\end{proof}
\begin{remark}\label{R:topo}
We remark that if $ \delta^\pm_\nu$, $\tilde\delta^\pm_\nu$, $L^\pm_\nu$, $\tilde L^\pm_\nu$ are contructed by Corollary \ref{C:existence-1} and define $\Delta$ by (\ref{E:defdelta}), (\ref{E:LV}), then (\ref{E:deltainv-1})-(\ref{E:deltainv-4}), (\ref{E:vr})  are valid by adapting the argument in the proof of Theorem 4.6 and 4.7 in \cite{S90}. 

\end{remark}
\section{The inverse scattering problem II}\label{S:inv-com}
We adapt the argument in \cite{ABT86} to complete the process of reconstructing the operator (\ref{E:lax}) in this section. That is, we need to find a proper gauge $b(x)\in K_0'$ which transforms $\breve{m}( x,\lambda)$ to the eigenfunction $m(x,\lambda)$ of (\ref{E:normalizedlax}). 

Write $a=\left(\begin{array}{cc} 0 & D\\ D &0\end{array}\right)$, and $D=\textsl{diag }(w_1,\cdots,w_n)$. Define
\begin{gather}
\vec x=(x_1,\cdots,x_n)=x(w_1,\cdots,w_n),\label{E:vecx}\\
 X=\sum_{i=1}^nx_ia_i,\ \textit{$a_i$ are defined by (\ref{E:ai}),}\label{E:vecX}\\
M(\vec x,\lambda)=\breve{m}(x,a,\lambda)=\breve{m}(x,\lambda).\label{E:M}
\end{gather}
\begin{lemma}\label{L:pdem'}
Suppose $\breve{m}( x,\lambda)$ is derived by Theorem \ref{T:q0} or \ref{T:qne0}. Then
\begin{equation}
\frac {\partial M}{\partial x_j}=\left[M,\lambda a_j+\frac 1\lambda\sigma_1(a_j)\right]+\frac 1\lambda\left(\sigma_1(a_j)-B_j(\vec x)\right)M-C_j(\vec x)M,\label{E:M'}
\end{equation}
with
\begin{gather}
	B_j(\vec x)\in \mathcal P_0\cap C^\infty,\quad C_j(\vec x)\in \mathcal K_0\cap C^\infty.\label{E:BC}
\end{gather}
\end{lemma}
\begin{proof} We are going to show that 
$\{\left(\partial_{x_j}+\textit{ad}(\lambda a_j+\frac 1\lambda\sigma_1(a_j))\right) M\}{M}^{-1}$ is holomorphic on $\mathbb C\backslash\{0\}$  
and bounded at $\infty$. Hence (\ref{E:M'}) follows immediately from the asymptotic expansions
\begin{eqnarray}
M(\vec x,\lambda)\to & 1+\sum_{k=1}^\infty M_k(\vec x)\lambda^{-k} &\textit{ as $|\lambda|\to \infty$,}\label{E:m'infty}\\
M(\vec x,\lambda)\to &\sum_{k=0}^\infty \tilde M_k(\vec x)\lambda^{k} &\textit{ as $|\lambda|\to 0$,}\label{E:m'0}
\end{eqnarray}
and the conditions (\ref{E:BC}) come from $\left(\left(\partial_{x_j}+\textit{ad}\left(\lambda a_j+\frac 1\lambda\sigma_1(a_j)\right)\right) \breve{m}\right){\breve{m}}^{-1} \in {\mathcal L}^{\sigma_0}$ by using 
$\breve{m}\in L^{\sigma_0}$, $\lambda a_j+\frac 1\lambda\sigma_1(a_j)\in{\mathcal L}^{\sigma_0}_+$, (\ref{E:regularityqo}), and (\ref{E:regularityqneo}). 

If $x\le 0$ and $\lambda\in\Omega^\pm_\nu$, by (1) in (\ref{E:brevem}), (\ref{E:vecx}), and (\ref{E:M}), then
\begin{eqnarray*}
	&&\{\left(\partial_{x_j}+\textit{ad}(\lambda a_j+\frac 1\lambda\sigma_1(a_j))\right) M\}{M}^{-1}\\
	=&& \{\left(\partial_{x_j}+\textit{ad}(\lambda a_j+\frac 1\lambda\sigma_1(a_j))\right) \eta^\pm_\nu(\vec x, \lambda) e^{-\lambda X-\frac 1\lambda\sigma_1(X)}(1+U^\pm_\nu(\lambda)\,)e^{\lambda X+\frac 1\lambda\sigma_1(X)}\}\\
	&&\cdot\left( \eta^\pm_\nu(\vec x, \lambda) e^{-\lambda X-\frac 1\lambda\sigma_1(X)}(1+U^\pm_\nu(\lambda)e^{\lambda X+\frac 1\lambda\sigma_1(X)}\right)^{-1}\\
	=&& \{\left(\partial_{x_j}+\textit{ad}(\lambda a_j+\frac 1\lambda\sigma_1(a_j))\right) \eta^\pm_\nu\,)\}( \eta^\pm_\nu)^{-1}.
\end{eqnarray*}
Here we have used  $\left(\partial_{x_j}+\textit{ad}(\lambda a_j+\frac 1\lambda\sigma_1(a_j))\right) e^{-(\lambda X+\frac 1\lambda\sigma_1(X))}(1+U)e^{\lambda X+\frac 1\lambda\sigma_1(X)} =0$. Using (2), (3), (4) in (\ref{E:brevem}) and the same argument, we will derive similar formula on other components. So  $\{\left(\partial_{x_j}+\textit{ad}(\lambda a_j+\frac 1\lambda\sigma_1(a_j))\right) M\}{M}^{-1}$ is regular on $\mathbb C\backslash \Sigma_a$ by the properties of $\eta^\pm_\nu$, $\tilde \eta^\pm_\nu$, $\rho^\pm_\nu$, and $\tilde \rho^\pm_\nu$ in Theorem \ref{T:q0}. 
Besides, by (\ref{E:tildecs}), (\ref{E:tildecsq0}), we derive
\begin{eqnarray*}
	&&\{\left(\partial_{x_j}+\textit{ad}(\lambda a_j+\frac 1\lambda\sigma_1(a_j))\right) M_+\}{M_+}^{-1}\\
=	&&\{\left(\partial_{x_j}+\textit{ad}(\lambda a_j+\frac 1\lambda\sigma_1(a_j))\right)( M_+-M_-)\}{M_+}^{-1}\\
&&+
\{(\partial_{x_j}+\textit{ad}(\lambda a_j+\frac 1\lambda\sigma_1(a_j)))M_-\}{M_+}^{-1}\\
=	&&\{\left(\partial_{x_j}+\textit{ad}(\lambda a_j+\frac 1\lambda\sigma_1(a_j))\right)M_-( e^{-(\lambda X+\frac 1\lambda\sigma_1(X))}V(\lambda)e^{\lambda X+\frac 1\lambda\sigma_1(X)}-1)\}{M_+}^{-1}\\
&&+
\{\left(\partial_{x_j}+\textit{ad}(\lambda a_j+\frac 1\lambda\sigma_1(a_j))\right)M_-\}{M_+}^{-1}
\\
=	&&\{\left(\partial_{x_j}+\textit{ad}(\lambda a_j+\frac 1\lambda\sigma_1(a_j))\right)M_-\}( e^{-(\lambda X+\frac 1\lambda\sigma_1(X))}V(\lambda)e^{\lambda X+\frac 1\lambda\sigma_1(X)}-1){M_+}^{-1}\\
&&+
\{\left(\partial_{x_j}+\textit{ad}(\lambda a_j+\frac 1\lambda\sigma_1(a_j))\right)M_-\}{M_+}^{-1}
\\
=	&&\{\left(\partial_{x_j}+\textit{ad}(\lambda a_j+\frac 1\lambda\sigma_1(a_j))\right)M_-\} {M_-}^{-1}
\end{eqnarray*}
Therefore, $\{\left(\partial_{x_j}+\textit{ad}(\lambda a_j+\frac 1\lambda\sigma_1(a_j))\right) M\}{M}^{-1}$ is continuous at $ \Sigma_a$. The uniform boundedness of $\{\left(\partial_{x_j}+\textit{ad}(\lambda a_j+\frac 1\lambda\sigma_1(a_j))\right) M\}{M}^{-1}$ at $\infty$ can be seen by (\ref{E:m'infty}).
\end{proof}

\begin{lemma}\label{L:CC}
The compatibility conditions of (\ref{E:M'}) are
\begin{equation}
\partial_{x_j}C_i-\partial_{x_i}C_j-\left[C_i,C_j\right]=
\left[B_i, a_j\right]-\left[B_j, a_i\right].\label{E:cc}
\end{equation}
\end{lemma}
\begin{proof}
Taking the derivative $\partial_{x_i}$ of (\ref{E:M'}), we have
\begin{eqnarray*}
	\partial_{x_i}\partial_{x_j}M
	&=&\left[[M,\lambda a_i+\frac 1\lambda\sigma_1(a_i)],\lambda a_j+\frac 1\lambda\sigma_1(a_j)\right]\\
	&&+\left[\frac 1\lambda\left(\sigma_1(a_i)-B_i\right)M-C_iM, \lambda a_j+\frac 1\lambda\sigma_1(a_j)\right]-(\frac 1\lambda\frac{\partial B_j}{\partial x_i}+\frac{\partial C_j}{\partial x_i})M
	\\
	&&+\left(\frac 1\lambda\left(\sigma_1(a_j)-B_j\right)-C_j\right)\left[M,\lambda a_i+\frac 1\lambda\sigma_1(a_i)\right]\\
	&&+\left(\frac 1\lambda\left(\sigma_1(a_j)-B_j\right)-C_j\right)\left(\frac 1\lambda\left(\sigma_1(a_i)-B_i\right)M-C_iM\right).
\end{eqnarray*}
Letting $|\lambda|\to \infty$ and applying (\ref{E:m'infty}), we obtain
\begin{equation}
C_j(\vec x)=\left[ M_1(\vec x),a_j\right].\label{E:m1}
\end{equation}
Hence 
\begin{eqnarray*}
&&\partial_{x_i}\partial_{x_j}M|_{\lambda=\infty}\\
=&&\left[\sigma_1(a_i)-B_i, a_j\right]-\left[C_i M_1, a_j\right]	-\frac{\partial C_j}{\partial x_i}-C_j\left[ M_1,a_i\right]+C_jC_i+\textit{symm. terms}\\
=&&-\left[B_i, a_j\right]-\frac{\partial C_j}{\partial x_i}+C_jC_i+a_jC_i M_1+C_ja_i M_1+\textit{symm. terms}\\
=&&-\left[B_i, a_j\right]-\frac{\partial C_j}{\partial x_i}+C_jC_i+\left[a_j,C_i\right] M_1-\left[a_i,C_j\right] M_1+\textit{symm. terms}\\
=&&-\left[B_i, a_j\right]-\frac{\partial C_j}{\partial x_i}+C_jC_i+\textit{symm. terms}.
\end{eqnarray*}
Here "`symm. terms"' denote terms which are symmetric with respect to $i$, $j$ and they may differ from each other. Also, in the above computation, we have also used
\[
\left[a_i, C_j \right], \ \left[[M,\lambda a_i+\frac 1\lambda\sigma_1(a_i)],\lambda a_j+\frac 1\lambda\sigma_1(a_j)\right]\ \textit{ are symmetric with respect to $i$, $j$}\]which follow from (\ref{E:m1}). Therefore, the compatibility conditions of (\ref{E:M'}) are
(\ref{E:cc}).
\end{proof}

\begin{lemma}\label{L:gauge}
Suppose either of the assumption in Theorem \ref{T:q0} or \ref{T:qne0} holds. 
Then there exists 
\begin{equation}
b(\vec x) \in K_0'\cap C^\infty, \label{E:gauge}
\end{equation}
such that
\begin{eqnarray}
b(x(w_1,\cdots,w_n))\to 1 &&\textit{as $x\to -\infty$,} \label{E:gauge-infty}\\
-bC_jb^{-1}+(\partial_jb)b^{-1}\in \mathcal S_0&&\textit{for $\forall j$.}\label{E:gauge-1}
\end{eqnarray}
\end{lemma}
\begin{proof} By (\ref{E:cc}), and (\ref{E:BC}), to prove (\ref{E:gauge}) and (\ref{E:gauge-1}), we need only to show 
\begin{equation}
\left[B_j, a_i\right]\in \mathcal S_0,\quad i\ne j.\label{E:bj-s1}
\end{equation}

First of all, let us claim
\begin{equation}
M(\vec x,\lambda)=\sigma_1({\tilde M_0}^{-1}M(\vec x,\frac 1\lambda)).\label{E:sigma2symmetry}
\end{equation}
Here $\tilde M_0$ is defined by (\ref{E:m'0}). 
The assertion can be proved by the Liouville's theorem and conditions (\ref{E:u-reality-5}), (\ref{E:tildecs}), (\ref{E:inv9})-(\ref{E:inv11}) and (\ref{E:brevem}). Therefore, taking the limits of (\ref{E:sigma2symmetry}) at $\lambda=0$, we derive
\begin{equation}
\tilde M_0\cdot \sigma_1({\tilde M_0})=1.\label{E:sigma2symmetry-1}
\end{equation}
Note $\tilde M_0\in K_0$ by $\breve{m}(x,\lambda)\in L^{\sigma_0}$. Hence 
\begin{equation}
\tilde M_0=\left(\begin{array}{cc}f_1 & 0 \\ 0 & f_2\end{array}\right),\  \textit{ with $f_i\in O(J)$} \label{E:m0'}
\end{equation}and $O(J)=\{x\in GL_n(\mathbb C)|Jx^*Jx=Jx^tJx=1\}$. Combining (\ref{E:sigma2symmetry-1}) and (\ref{E:m0'}), we have $f^2_1=1$. Therefore the minimal polynomial of $f_1$ must be a divisor of $\lambda^2-1$. If $J=I$, then (\ref{E:m0'}) implies $f_1$ is always diagonalizable. If $J\ne I$, and $|V-1|_{L^2}<<1$, by continuity, $f_1$ is  diagonalizable, too. Hence the minimal polynomial of $f_1$ must be $\lambda-1$. Therefore, we conclude 
\begin{equation}
\tilde M_0\in K_0'.\label{E:m0'-k1}
\end{equation}

Now let us equate the $\lambda^{-1}$-terms of (\ref{E:M'}) at $\lambda=0$, we obtain
\begin{equation}
\tilde M_0\sigma_1(a_j)-B_j\tilde M_0=0,\label{E:bj}
\end{equation}
Plugging (\ref{E:m0'-k1}) into (\ref{E:bj}) and solving for $B_j$, one can justify (\ref{E:bj-s1}). 
\end{proof}

\begin{theorem}\label{T:inv}
Suppose either of the assumption in Theorem \ref{T:q0} or \ref{T:qne0} holds. Let 
\begin{eqnarray}
\Psi(x,\lambda)&=&b(\vec x)M(\vec x,\lambda)e^{-(\lambda X+\frac 1\lambda\sigma_1(X))}
\label{E:definepsi}
\end{eqnarray}
Here $x$, $\vec x$, $X$, $M$ satisfy (\ref{E:vecx})-(\ref{E:M}). Then
\begin{gather*}
\frac{\partial \Psi}{\partial x}=-\lambda bab^{-1}\Psi-\frac 1\lambda\sigma_1(bab^{-1})\Psi-v\Psi, 
\end{gather*}
with
\begin{eqnarray}
	b(x)&=&b(x(w_1,\cdots,w_n))\in K_0'\cap C^\infty,\nonumber\\
v(x)&=&\sum w_j(bC_jb^{-1}-(\partial_jb)b^{-1})(x(w_1,\cdots,w_n))\in\mathcal S_0\cap C^\infty,\label{E:C}
\end{eqnarray}
where $C_j$, $b(x(w_1,\cdots,w_n))$ are defined by Lemma \ref{L:pdem'} \ref{L:gauge}, respectively.
\end{theorem}
\begin{proof}
Let $\Phi(x,\lambda)=M(\vec x,\lambda)e^{-(\lambda X+\frac 1\lambda\sigma_1(X))}$. Then
\begin{eqnarray*}
\frac{\partial\Phi}{\partial {x}} &=&\sum_{j=1}^n w_j\partial_{x_j}\Phi\\
&=&\sum_{j=1}^n w_j(-\lambda a_j-\frac 1\lambda B_j-C_j)\Phi
\end{eqnarray*}
by (\ref{E:M'}). Therefore, using formula (\ref{E:definepsi}),
\begin{eqnarray}
\frac{\partial \Psi}{\partial x}
&=&(-\lambda b(\sum_{j=1}^n w_j a_j)b^{-1} -\frac 1\lambda b(\sum_{j=1}^n w_j B_j)b^{-1} +\sum w_j(-bC_jb^{-1}+(\partial_jb)b^{-1}))\Psi	\nonumber\\
&=&(-\lambda bab^{-1} -\frac 1\lambda b(\sum_{j=1}^n w_j B_j)b^{-1} -v)\Psi.	\label{E:gaugepsi}
\end{eqnarray}
with
$
	v= \sum w_j(bC_jb^{-1}-(\partial_jb)b^{-1})\in\mathcal S_0
	$ 
by Lemma \ref{L:gauge}. The proof reduces to showing the $\sigma_1(b(\sum_{j=1}^n w_j B_j)b^{-1})=bab^{-1}$. 
Define
\begin{eqnarray}
	A&=& -\sum_{j=1}^n w_j ba_jb^{-1}\, dx_j,\label{E:A}\\ 
	B&=& -\sum_{j=1}^n w_jb B_jb^{-1}\,dx_j,\label{E:B}\\
	C&=& \sum_{j=1}^n w_j(-bC_jb^{-1}+(\partial_jb)b^{-1})\,dx_j\label{E:C-1}
\end{eqnarray}
and write (\ref{E:gaugepsi}) as
\[d\Psi=\lambda A\Psi+\frac 1\lambda B\Psi+C\Psi.
\]So $d^2\Psi=0$. This implies
\begin{gather*}
	dA+A\wedge C+C\wedge A=0,\\
	dB+B\wedge C+C\wedge B=0.
\end{gather*}
Thus
\begin{equation}
d(A-\sigma_1(B))+(A-\sigma_1(B))\wedge C+C\wedge (A-\sigma_1(B))=0\label{E:symm-1}
\end{equation}
by (\ref{E:C}), (\ref{E:C-1}).
Along the direction $x(w_1,\cdots, w_n)$, as $x\to -\infty$, we obtain
\begin{eqnarray}
&&A-\sigma_1(B)\nonumber\\
&=& \sum_{j=1}^n -w_jba_jb^{-1}\,dx_j+w_j\sigma_1(bB_jb^{-1})\,dx_j\nonumber\\
&=&\sum_{j=1}^n w_jb(-a_j+	b^{-1}\sigma_1(bB_jb^{-1})b)b^{-1}\, dx_j\nonumber\\
&\to & 0\label{E:infty-inv}
\end{eqnarray}
by (\ref{E:regularityqo}), (\ref{E:regularityqneo}), (\ref{E:gauge-infty}), (\ref{E:bj}). Consequently, (\ref{E:symm-1}), (\ref{E:infty-inv}) yield
\[A=\sigma_1(B).\]
\end{proof}

\begin{remark}\label{R:sigmar-2}
For $i\in\{0,\,1,\,\cdots,\,n-1\}$,  replacing $\sigma_1$ by $\sigma_i$ defined in Remark \ref{R:sigmar}, we can solve the associated inverse problem  by analogy.
\end{remark}

\section{The Cauchy problem}\label{S:cauchy}
\begin{theorem}\label{T:cauchy}
Let $a,\ \tilde a\in\mathcal A$ be constant oblique directions, $b_0(x)\in K'_0$, $v_0(x)\in\mathcal S_0$.  Suppose  
$ b_0-1$, $v_0$
and their derivatives are rapidly decreasing as $|x|\to \infty$ and 
 the set $Z$ in Theorem \ref{T:existence} is a finite set contained in $\mathbb C\backslash\Sigma_a$. In case of $q\ne 0$, $(b_0, v_0)$ satisfies additionally the small data constraint $|b_0-1|_{L^1_1}+|v_0|_{L_1}<<1$. Then the Cauchy problem of the twisted flow
\begin{eqnarray*}
&&\left[\partial_x+bab^{-1}\lambda+v+\sigma_1(bab^{-1})\lambda^{-1}, \partial_t+\sum_{s=1}^{2j+1} Q_s \lambda^s+Q_0+\sum_{s=1}^{2j+1} \sigma_1(Q_s)\lambda^{-s}\right]=0,\\
&&{}\\
&&b(x,0)=b_0,\quad v(x,0)=v_0
\end{eqnarray*}
admits a smooth solution for $x\in\mathbb R$, $t\ge 0$. Here $Q_s (x,t)$, $0\le s\le 2j+1$, are defined by (\ref{E:coeffM}), Lemma \ref{L:Acoeff}.
\end{theorem}
\begin{proof}
We first apply Theorem \ref{T:existence}, Corollary \ref{C:existence-1}, Definition \ref{D:scattering}  and 
Theorem \ref{T:sc} to obtain the scattering data $(U^\pm_{\nu,0}(\lambda), V_0(\lambda))$ for the potential $(b_0, v_0)$. Define
\begin{eqnarray}
U^\pm_{\nu}(\lambda,t)&=&e^{-t(\lambda^{2j+1} \tilde a+\frac 1{\lambda^{2j+1}}\sigma_1(\tilde a))}U^\pm_{\nu,0} (\lambda)
e^{t(\lambda^{2j+1} \tilde a+\frac 1{\lambda^{2j+1}}\sigma_1(\tilde a))},\ \textit{$\lambda\in\Omega_\nu^\pm$}\label{E:cauchy1}\\
V(\lambda,t)&=&e^{-t(\lambda^{2j+1} \tilde a+\frac 1{\lambda^{2j+1}}\sigma_1(\tilde a))}V_0(\lambda) e^{t(\lambda^{2j+1} \tilde a+\frac 1{\lambda^{2j+1}}\sigma_1(\tilde a))},\ \textit{$\lambda\in\Sigma_a$}.\label{E:cauchy2}
\end{eqnarray}
Hence $(U^\pm_{\nu}, V)$ satisfies (\ref{E:ana1})-(\ref{E:u-reality-5}) and $|V-1|_{L_\infty}<<1$ if $q\ne 0$ \cite{BC84}. Thus one can apply Theorem \ref{T:q0}, \ref{T:qne0}, and  \ref{T:inv} to construct $M(x,t,\lambda)$, $b(x,t)$, $v(x,t)$. Let 
\begin{gather*}
m(x,t,\lambda)=b(x,t)M(x,t,\lambda),\\
\Psi(x,t,\lambda)=m(x,t,\lambda)e^{-x(\lambda a+\frac 1\lambda\sigma_1(a))-t(\lambda^{2j+1} \tilde a+\frac 1{\lambda^{2j+1}}\sigma_1(\tilde a))}.
\end{gather*}
Then 
\begin{gather*}
\frac{\partial \Psi}{\partial x}=-\lambda bab^{-1}\Psi-\frac 1\lambda\sigma_1(bab^{-1})\Psi-v\Psi, \\
b\in K_0'\cap C^\infty,\quad v\in\mathcal S_0\cap C^\infty.
\end{gather*}
So $\Psi(x,t,\lambda)\in L^{\sigma_0}_+$ and $\frac{\partial \Psi}{\partial t}\Psi(x)^{-1}\in\mathcal L^{\sigma_0}_+$. Moreover,
\begin{eqnarray*}
&&\frac{\partial \Psi}{\partial t}\Psi(x)^{-1}\\
=&&\left\{\frac{\partial }{\partial t}\left(m(x,t,\lambda)e^{-x(\lambda a+\frac 1\lambda\sigma_1(a))-t(\lambda^{2j+1} \tilde a+\frac 1{\lambda^{2j+1}}\sigma_1(\tilde a))}\right)\right\}\Psi^{-1}\\
=&&\left\{\left[\frac{\partial m}{\partial t}-m(\lambda^{2j+1} \tilde a+\frac 1{\lambda^{2j+1}}\sigma_1(\tilde a))\right]e^{-x(\lambda a+\frac 1\lambda\sigma_1(a))-t(\lambda^{2j+1} \tilde a+\frac 1{\lambda^{2j+1}}\sigma_1(\tilde a))}\right\}\\
&&\cdot e^{x(\lambda a+\frac 1\lambda\sigma_1(a))+t(\lambda^{2j+1} \tilde a+\frac 1{\lambda^{2j+1}}\sigma_1(\tilde a))}m^{-1}\\
=&&\frac{\partial m}{\partial t}m^{-1}-m\left(\lambda^{2j+1} \tilde a+\frac 1{\lambda^{2j+1}}\sigma_1(\tilde a)\right)m^{-1}\\
=&&-\hat\pi_+(m J_{\tilde a,2j+1}m^{-1})\\
=&&-\left(\sum_{s=1}^{2j+1} Q_s \lambda^s+Q_0+\sum_{s=1}^{2j+1} \sigma_1(Q_s)\lambda^{-s}\right)
\end{eqnarray*}
by  (\ref{E:coeffM}), Lemma \ref{L:gaugeM}-\ref{L:Acoeff}.
 
\end{proof}

\begin{theorem}\label{T:GSHGE}
Let $a=\textit{diag }(w_1,\cdots,w_n)\in\mathcal A$ be a constant oblique direction,  $b_0(x)\in K'_0$, $v_0(x)\in\mathcal S_0$,  
$ b_0-1$, $v_0$
and their derivatives are rapidly decreasing as $|x|\to \infty$ and 
 the set $Z$ in Theorem \ref{T:existence} is a finite set contained in $\mathbb C\backslash\Sigma_a$. In case of $q\ne 0$, $(b_0, v_0)$ satisfies additionally the small data constraint $|b_0-1|_{L^1_1}+|v_0|_{L_1}<<1$. 
 Then there exists a solution to the $1$-dimensional twisted $\frac{O(J,J)}{O(J)\times O(J)}$-system (\ref{E:1Dsystem1}) satisfying
\begin{eqnarray*}
&&b(x(w_1,\cdots,w_n)\,)=b_0(x),\quad \sum_{k=1}^n v_k\,(x(w_1,\cdots,w_n)\,) =v_0(x).
\end{eqnarray*}
\end{theorem}
\begin{proof} We first apply Theorem \ref{T:existence}, Corollary \ref{C:existence-1}, and 
Theorem \ref{T:sc} to obtain the scattering data $(U^\pm_{\nu,0}(\lambda), V_0(\lambda))$ for the potential $(b_0, v_0)$.  Then we apply Theorem \ref{T:q0}, \ref{T:qne0}, \ref{T:inv} to construct $C_k(x_1,\cdots,x_n)$, $b(x_1,\cdots,x_n)$, $v_k(x_1,\cdots,x_n)$, $M(x_1,\cdots,x_n)$, such that $\Psi(x_1,\cdots,x_n)=b(x_1,\cdots,x_n)M(x_1,\cdots,x_n)e^{-\lambda\sum x_ka_k -\frac 1\lambda\sum\sigma_1(x_ka_k)}$ satisfies
\begin{gather*}
\frac{\partial \Psi}{\partial x_k}
=-\lambda ba_kb^{-1}\Psi  -v_k\Psi-\frac 1\lambda\sigma_1(ba_kb^{-1})\Psi,	
\end{gather*}
with $a_k$ defined by (\ref{E:ai}), $v_k=bC_kb^{-1}-(\partial_kb)b^{-1}$, and
\begin{eqnarray*}
	b_0(x)&=&b(x(w_1,\cdots,w_n))\in K_0'\cap C^\infty,\nonumber\\
v_0(x)&=& \sum_{k=1}^n v_k(x(w_1,\cdots,w_n))\in\mathcal S_0\cap C^\infty.
\end{eqnarray*} 

\end{proof}

\noindent\textbf{\emph{Acknowledgements}} 

{The first author was partially supported by NSFC grant No. 10971111 and the second author was  partially supported by NSC 99-2115-M-001-002. The second author would like to thank Professor C. L. Terng for introducing twisted hierarchies, Professor V. Sokolov for providing important references, and Professor M. Guest for his interest and encouragement.}

\end{document}